\documentclass[11pt]{article}
\usepackage[a4paper,hmargin=1.0in,vmargin=1.0in]{geometry}
\usepackage{amsmath,amsthm,amssymb,setspace,sectsty,bbm,mathtools,graphicx}
\usepackage[usenames,dvipsnames]{xcolor}
\usepackage[round]{natbib}
\usepackage[linktocpage=true,pagebackref=true,colorlinks,allcolors=blue,bookmarks=true,bookmarksopen,bookmarksnumbered]{hyperref}

\newtheoremstyle{DStheorem}
  {\topsep}
  {\topsep}
  {\itshape}
  {0pt}
  {\scshape}
  {.}
  { }
  {\thmname{#1}\thmnumber{ #2}\thmnote{ (#3)}}
\theoremstyle{DStheorem}
\newtheorem{theorem}{Theorem}[section]
\newtheorem{lemma}[theorem]{Lemma}

\newtheorem{claim}[theorem]{Claim}

\let\oldproofname=\proofname
\renewcommand{\proofname}{\rm\sc{\oldproofname}}

\newcommand{\MyAbove}[2]{\genfrac{}{}{0pt}{}{#1}{#2}}
\newcommand{\bs}[1]{\boldsymbol{#1}}
\newcommand{\bb}{\mathbb}
\newcommand{\eps}{\epsilon}
\newcommand{\mystick}{\mathrm{stick}}
\newcommand{\myquit}{\mathrm{quit}}
\newcommand{\myavg}{\mathrm{average}}
\newcommand{\mymark}{{\cal M}}
\newcommand{\myunmark}{\bar{\cal M}}
\newcommand{\myco}{\mathrm{CO}}
\newcommand{\bbN}{\mathbbm{N}}
\newcommand{\bbInd}{\mathbbm{1}}
\newcommand{\rev}{{\cal R}}
\newcommand{\pr}[1]{{\rm Pr} \left[ #1 \right]}
\newcommand{\prpar}[1]{{\rm Pr} [ #1 ]}
\newcommand{\prsub}[2]{{\rm Pr}_{#1} \left[ #2 \right]}
\newcommand{\prparsub}[2]{{\rm Pr}_{#1} [ #2 ]}
\newcommand{\ex}[1]{{\mathbb E} \left[ #1 \right]}
\newcommand{\expar}[1]{{\mathbb E} [ #1 ]}
\newcommand{\exsub}[2]{{\mathbb E}_{#1} \left[ #2 \right]}

\newcommand{\poly}{\mathrm{poly}}

\sectionfont{\large} \subsectionfont{\normalsize}
\allowdisplaybreaks
\makeindex

\begin{document}

\begin{titlepage}

\title{Near-Optimal Adaptive Policies for \\
Serving Stochastically Departing Customers}
\author{%
Danny Segev\thanks{Department of Statistics and Operations Research, School of Mathematical Sciences, Tel Aviv University, Tel Aviv 69978, Israel. Email: {\tt segevdanny@tauex.tau.ac.il}. The work on this project is supported by Israel Science Foundation grant 1407/20.}%
}
\date{}
\maketitle

\thispagestyle{empty}

\begin{abstract}
We consider a multi-stage stochastic optimization problem originally introduced by \cite{CyganEGMS13}, studying how a single server should prioritize stochastically departing customers. In this setting, our objective is to determine an adaptive service policy that maximizes the expected total reward collected along a discrete planning horizon, in the presence of customers who are independently departing between one stage and the next with known stationary probabilities. In spite of its deceiving structural simplicity, we are unaware of non-trivial results regarding the rigorous design of optimal or truly near-optimal policies at present time.

Our main contribution resides in proposing a quasi-polynomial-time approximation scheme for adaptively serving impatient customers. Specifically, letting $n$ be the number of underlying customers, our algorithm identifies in $O( n^{ O_{ \eps }( \log^2 n ) } )$ time an adaptive service policy whose expected reward is within factor $1 - \eps$ of the optimal adaptive  reward. Our method for deriving this approximation scheme synthesizes various stochastic analyses in order to investigate how the adaptive optimum is affected by alteration to several instance parameters, including the reward values, the departure probabilities, and the collection of customers itself.
\end{abstract}


\end{titlepage}

\thispagestyle{empty}
\tableofcontents

\newpage
\setcounter{page}{1}

\section{Introduction}

One fundamental issue in classical and present-day  stochastic systems is that of providing timely service to queued customers under high capacity utilization. Since  customers are impatient in various ways, it is important for decision-makers to devise prioritization policies for allocating limited capacity. In this context, customers' impatience may manifest in a multitude of interrelated phenomena, such as balking, where customers could choose to avoid joining a crowded queue in the first place, or reneging, where they may join a given queue and subsequently leave without receiving service, often due to frustration and dissatisfaction, or due to having access to alternative options. These well-documented behavior patterns were shown to have major economic consequences in several operational settings, including data centers, healthcare, ride-hailing, call centers, and routing. To better appreciate these application domains, we refer the reader to a number of surveys on this topic \citep{harris1987modeling,aksin2007modern,wang2010queueing,armony2015patient,batt2015waiting,yan2020dynamic} and to the references therein.

In this paper, we revisit one of the structurally simplest multi-stage stochastic optimization problems representing customer impatience, for which computationally efficient methods are still mostly unknown. Specifically, we study a model formulation originally introduced by \cite{CyganEGMS13}, who examined how a single server should prioritize heterogeneous stochastically departing customers. Somewhat informally, at the beginning of the planning horizon, $n$ customers are waiting in a queue to be potentially served along a sequence of discrete stages. In each such stage, the decision-maker picks one of the currently-available customers, say $i \in [n]$, to be served next, and collects a corresponding value $r_i$ as a reward. Subsequently, each of the remaining customers $j\neq i$ may independently depart from the system with probability $p_j$ prior to the next stage. This process repeats itself until all customers have either been served or departed. Our objective is to determine an adaptive service policy that maximizes the expected total reward collected along the entire planning horizon. Section~\ref{subsec:model_desc} presents a complete mathematical description of this model and delves deeper into its system dynamics. 

As explained in Section~\ref{subsec:known-results}, in spite of its deceiving structural simplicity, we are unaware of any non-trivial results regarding the rigorous design of optimal or truly near-optimal policies in this context. To provide some intuition around the technical obstacles hiding under the surface, it is worth noting that, for the purpose of maximizing the expected total reward, one should  carefully account for customer heterogeneity in devising a prioritization policy, i.e., how much we value serving different customers versus how patient they are. All other things being equal, it would be preferable to prioritize either highly impatient customers (with large departure probabilities $p_i$-s) or high-value customers (with large rewards $r_i$-s). However, since these parameters could be completely unrelated, impatient customers generally do not coincide with high-value ones, meaning that the design of optimal or near-optimal  adaptive service policies requires resolving a complex tradeoff between these attributes. As we proceed to explain next, these issues can easily be circumvented by means of dynamic programming, with a state space description that keeps track of the precise identity of each available customer at any point in time. However, this approach suffers from the rather obvious curse of dimensionality, being exponential in the number of underlying customers.   

In the remainder of this exposition, we dedicate Section~\ref{subsec:model_desc} to describing the above-mentioned model formulation in greater detail. Next, in Section~\ref{subsec:known-results}, we review known results in this context and highlight our motivating research questions. Section~\ref{subsec:contributions} provides a formal account of our main results and sheds light on selected technical ideas. Finally, we briefly review several lines of closely related literature in Section~\ref{subsec:related}.

\subsection{Model description} \label{subsec:model_desc}

\paragraph{Input parameters and service policies.} We are given a finite collection of customers, which will be referred to as $1, \ldots, n$. Each customer $i \in [n]$ is associated with a non-negative reward $r_i$ as well as with a departure probability of $p_i$ whose precise meaning will be explained shortly. Given these parameters, the random process we consider evolves along a sequence of discrete stages according to system dynamics defined through adaptive service policies. To formalize this notion, for any stage $t \in \bbN$, let us make use of ${\cal A}_t$ to denote the set of customers available at the beginning of stage $t$; we refer to every possible pair $(t,{\cal A}_t) \in \bbN \times 2^{[n]}$ as a state. An adaptive service policy is simply a function ${\cal S} : \bbN \times 2^{[n]} \to [n]$ that, given any state $(t,{\cal A}_t)$ with ${\cal A}_t \neq \emptyset$, decides on a single customer ${\cal S}({t,{\cal A}_t})$ to be served at that time, picked out of the currently available set of customers ${\cal A}_t$.

\paragraph{System dynamics.} At the beginning of stage 1, all customers are  present in the system, meaning that ${\cal A}_1 = [n]$. Then, at each stage $t$, when the set ${\cal A}_t$ of currently available customers is non-empty, the following sequence of steps occurs:
\begin{enumerate}
    \item {\em Service step}: We first let the policy ${\cal S}$ pick one of the available customers, ${\cal S}({t,{\cal A}_t})$. Once we provide service to the latter, our cumulative total reward is incremented by the corresponding reward, $r_{ {\cal S}({t,{\cal A}_t})}$, and this customer leaves the system.

    \item {\em Departure step}: Out of the remaining set, ${\cal A}_t \setminus \{ {\cal S}({t,{\cal A}_t}) \}$, each customer $i$ independently decides whether to depart from the system (with probability $p_i$) or to keep waiting to be potentially served in subsequent stages (with probability $1 - p_i$).
    
    \item {\em State update:} Moving forward, ${\cal D}_t$ will stand for the random subset of customers who decided to depart, given the current state $(t,{\cal A}_t)$ and the identity ${\cal S}({t,{\cal A}_t})$ of the customer who has just been served. As such, we proceed to stage $t+1$ with the set of customers
    \[ {\cal A}_{t+1} ~~=~~ {\cal A}_t \setminus \left( \{ {\cal S}({t,{\cal A}_t}) \} \cup {\cal D}_t \right) \ . \]
\end{enumerate}
This random process terminates as soon as we arrive at a state $(t, {\cal A}_t)$ in which the available set of customers ${\cal A}_t$ is empty.

\paragraph{Reward function and objective.} For an adaptive service policy ${\cal S} : \bbN \times 2^{[n]} \to [n]$ and a state $(t, {\cal A}_t)$, let $\rev_{\cal S}(t, {\cal A}_t)$ be the expected cumulative reward attained by the policy ${\cal S}$ along the process in question, starting at state $(t, {\cal A}_t)$. It is easy to verify that, when ${\cal A}_t \neq \emptyset$, this function can be recursively expressed as
\begin{equation} \label{eqn:recursive_rev}
\rev_{\cal S}(t, {\cal A}_t) ~~=~~ r_{ {\cal S}({t,{\cal A}_t}) } + \ex{\rev_{\cal S}(t+1, {\cal A}_t \setminus ( \{ {\cal S}({t,{\cal A}_t}) \} \cup {\cal D}_t ) ) } \ ,
\end{equation}
where the expectation above is taken over the randomness in generating the set of departing customers ${\cal D}_t$. In the opposite case where ${\cal A}_t = \emptyset$, we clearly have $\rev_{\cal S}(t, {\cal A}_t) = 0$. Our goal is to compute an adaptive service policy ${\cal S}$ whose expected total reward $\rev( {\cal S} )$ is maximized. The latter function stands for the expected reward with regards to the initial state where all customers are present at the beginning of stage 1, meaning that $\rev( {\cal S} ) = \rev_{\cal S}(1,{\cal A}_1)$. 

\subsection{Known results and open questions} \label{subsec:known-results}

As further elaborated below, the original work of \cite{CyganEGMS13} presented a host of algorithmic results, providing constant-factor approximation guarantees through easy-to-implement service policies. Interestingly, the authors established a clear separation between the type of approximations that can be attained relative to two distinct benchmarks:
\begin{enumerate}
    \item The {\em offline optimum}, where our clairvoyant adversary knows in advance how customers' departure decisions will be realized. In other words, this benchmark takes the perspective of competitive analysis.
    
    \item The {\em adaptive optimum}, namely, the maximum-possible expected total reward collected by any adaptive service policy. This benchmark corresponds to the optimal value of the dynamic program~\eqref{eqn:recursive_rev}.
\end{enumerate}

\paragraph{The online benchmark.} In regard to this measure, \cite{CyganEGMS13} observed that constant-factor competitiveness can be established by exploiting connections to the vertex-weighted online matching problem. More precisely, one can re-interpret customers as offline vertices on one side of a bipartite graph; similarly, stages are analogous to vertices that arrive in online manner on the other side of this graph. When stage $t$ arrives, its corresponding vertex will be connected by newly added edges to all customers who have not departed yet. Given this reduction, a competitive ratio of $1 - \frac{ 1 }{ e } \approx 0.632$ for serving impatient customers follows from the work of \cite{AggarwalGKM11}, who attained this ratio for vertex-weighted online matching. Reciprocally, classical families of hard instances for the online matching problem cannot be captured in the opposite direction, due to the specific randomization behind customers' departure decisions. Hence, it is still unknown whether $(1 - \frac{ 1 }{ e })$-competitiveness can be breached  for this problem. In any case, such improvements would not be dramatic, in light of an upper bound of roughly $0.648$ on the achievable competitive ratio due to \cite{CyganEGMS13}.

\paragraph{The adaptive benchmark.} Moving on to comparisons against the adaptive optimum, which is the primary focus of our work, the main result of \cite{CyganEGMS13} resides in devising a slightly-adaptive $0.709$-approximate service policy, which is polynomially computable. In fact, the latter approximation is attained relative to an LP relaxation, which forms an upper bound on the adaptive optimum. In a nutshell, starting from an optimal fractional solution, their policy first computes a  customer-to-stage assignment by means of randomized rounding. Then, in each stage, the highest-reward customer assigned to this stage is served; moreover, when no customer was served in the previous stage, the second-highest-reward customer is served as well. Both choices are, of course, conditional on the current availability of these customers, and can eventually be implemented as a valid policy, serving at most one customer per stage. Importantly, this result implies that the best-achievable polynomial-time approximation is strictly better than the best-possible competitive ratio. We mention in passing that, even though improved guarantees may potentially be extracted along similar lines, there is still an inherent constant-factor gap between the adaptive optimum and the LP-based one. In support of this claim,  \cite{CyganEGMS13} derived an upper bound of $1 - \frac{1}{2e} \approx 0.816$ on this gap, showing that the LP benchmark is too loose to devise service policies whose expected reward is arbitrarily close to optimal.

\paragraph{Research questions.} To our knowledge, since the original work of \cite{CyganEGMS13}, improved approximation ratios for this setting have yet to be established. In fact, quoting Anupam Gupta's talk at the Simons Institute \citeyearpar{Gupta16simons}, little is known about the computational aspects of the model in question, including whether this problem is PSPACE-hard or not. Hence, our work is motivated by the following open questions:
\begin{quote}
    \begin{itemize}
        \item {\em Approximability.} Can the adaptive optimum be approximated within any degree of accuracy, say by allowing slightly super-polynomial time algorithms?
        
        \item {\em Structure of value function.} Can we rigorously study how the adaptive optimum is affected by various alterations to the instance structure and input parameters? Could such properties be harnessed for computational purposes?
        
        \item {\em Algorithmic ideas.} Due to the inherent limitations of current LP-based methods, are there alternative techniques that would be useful in this context?
    \end{itemize}
\end{quote}

\subsection{Contributions and techniques} \label{subsec:contributions}

The main contribution of this paper resides in proposing a quasi-polynomial-time approximation scheme (QPTAS) for adaptively serving impatient customers. In other words, for any accuracy level $\eps > 0$, our algorithm identifies in slightly super-polynomial time an adaptive service policy whose expected reward is within factor $1 - \eps$ of the best-possible expected reward attainable by any such policy. The precise performance guarantees of this approach are described in the next theorem, where to avoid cumbersome expressions, we make use of $O_{ \eps }$ to suppress polynomial dependencies on $\frac{ 1 }{ \eps }$, meaning that $O_{ \eps }( f(n) ) = \poly( \frac{ 1 }{ \eps } ) \cdot f(n)$.

\begin{theorem} \label{thm:main_result}
For any accuracy level $\eps \in (0, \frac{ 1 }{ 4 })$, there is a deterministic $O( n^{ O_{ \eps }( \log^2 n ) } )$-time algorithm for computing an adaptive service policy whose expected total reward is within factor $1 - \eps$ of optimal. 
\end{theorem}

\paragraph{Main technical ideas.} Our method for deriving the above-mentioned approximation scheme synthesizes various stochastic analyses, eventually leading to an approximate dynamic programming formulation. At a high level, we investigate how the adaptive optimum is affected by various alteration to several instance parameters, including the reward values, the departure probabilities, and the collection of customers itself. As further explained below, our main findings show that, with a negligible loss in optimality, we can restrict attention to instances satisfying the following properties:
\begin{enumerate}
    \item {\em Reasonably patient customers:} The departure probabilities $\{ p_i \}_{i \in [n]}$ are $\poly(n,\frac{1}{\eps})$-bounded away from both $0$ and $1$.
    
    \item {\em Few rewards:} The collection of rewards $\{ r_i \}_{i \in [n]}$ consists of $O_{\eps}( \log n )$ distinct values.
    
    \item {\em Few probabilities:} The collection of departure probabilities $\{ p_i \}_{i \in [n]}$ consists of $O_{\eps}( \log n )$ distinct values.
\end{enumerate}
Based on properties~2 and~3, we argue that an optimal adaptive service policy can be exactly computed in $O( n^{ O_{ \eps }( \log^2 n ) } )$ time by means of dynamic programming. 

To keep this introduction concise, Section~\ref{sec:overview} is devoted to a step-by-step technical overview of our approach, whereas the remainder of the current section is meant to instill some preliminary motivation. Specifically, in Section~\ref{subsec:reduce_average}, we enforce property~1 by partitioning customers into three classes: ``quitters'', ``stickers'', and ``average''. Informally, quitters are very likely to leave within a single stage, whereas stickers stay up until very late stages with a high probability. As intuition would suggest, we prove the existence of near-optimal policies that are class-ordered, meaning that they start by potentially serving a single quitter in stage~1, while deferring the service of all stickers until the latest possible stages. Consequently, we reduce arbitrarily structured problem instances to those solely consisting of average customers. 

Next, in Sections~\ref{subsec:round_values} and~\ref{subsec:round_prob}, we respectively enforce properties~2 and~3 by appropriately rounding the customer rewards and departure probabilities. As it turns out, accounting for the effects of reward-rounding is rather straightforward, since the overall customer departure process remains probabilistically identical. That said, due to considering rounded departure probabilities, we are in fact altering the latter process, and rigorously analyzing the effects of this alteration appears to be a challenging question. Our main technical contribution resides in proposing a sensitivity analysis, proving that even though we are considering a modified process, we could be losing only an $O(\eps)$-fraction of the optimal expected reward. This result is formalized in Section~\ref{subsec:rev_preserve}, with its complete derivation presented in Section~\ref{sec:proof_thm_rev_diff_direction}. Technically speaking, our analysis is based on probabilistic coupling ideas, which enable us to compare the reward processes for two distinct instantiations of the departure probabilities. In particular, for any policy in the original departure process, we carefully construct a corresponding randomized policy in the modified process that yields comparable expected total rewards. The latter policy employs a simulated state-system for the original process, conditional on its own history of observed departures, to ``align'' the two processes and to facilitate the comparison of their expected rewards. We believe that the technical ideas behind our probabilistic coupling may be of broader interest.

\subsection{Related work} \label{subsec:related}

We conclude this section by positioning our contributions within the literature on service operations with impatient customers. To this end, it is instructive to briefly mention related work at the intersection of stochastic modeling and optimization, and to explain the uniqueness of our results in relation to existing analytical findings.

\paragraph{Queuing models with reneging.}  The issue of reneging (or abandonments) has motivated long-standing lines of research in queuing theory, tracing back to the work of \citet{haight1959queueing}. In the simplest setting, one considers an M/M/1 queue with reneging, where customers independently wait up to an exponentially distributed patience random variables and depart if they have not been served thus far. Much of the existing literature along these lines focuses on describing how simple priority rules (or service disciplines) perform. For example, under the first-come-first-served rule, the steady-state distribution of queue length, waiting times, and customer abandonments can be derived in closed form, even for more general queuing models \citep{ancker1963some}. Hence, various performance metrics can be asymptotically analyzed, depending on the amount of extra capacity endowed to the system \citep{garnett2002designing,zeltyn2005call}. The dynamic control problem, where the objective is to choose how to prioritize over different classes of customers waiting, has also received a great deal of attention.  Without abandonments, the celebrated $c\mu$ priority rule is known to minimize expected holding costs in many scenarios \citep{klimov1975time,baras1985two,buyukkoc1985cmu,van1995dynamic}; here, customer classes differ in their holding costs, indexed by $c$, and in their mean service times, indexed by $\mu$. That said, the optimality of index policies breaks in the presence of customer abandonments. Most existing near-optimality results in this context are derived under appropriately defined asymptotic regimes, such as a fluid limit with an overload condition \citep{atar2010cmu}, or a Brownian approximation \citep{AtaT13}. In addition, these models differ from our problem formulation in that their goal is to minimize a cost function that blends waiting times and customer abandonments, rather than to maximize total welfare from the valuations of eventually-served customers.

\paragraph{Dynamic matching.} Concurrently, serving impatient customers is a fundamental issue in matching markets. These settings are generally modelled through a stochastic process that specifies the arrivals of agents of different types, with an underlying compatibility graph that describes the feasible matches between agent types and their potential rewards. A centralized platform decides on how to match these agents on-the-fly. In the absence of abandonments, one wishes to minimize average waiting times, while achieving the maximum-possible throughput \citep{tsitsiklis2017flexible,anderson2017efficient}. Optimal policies in non-asymptotic regimes are notoriously difficult to characterize, with only a few exceptions that further exploit simple graph structures \citep{cadas2019optimal}. Adding an important layer of complexity, the presence of stochastic abandonments, which endow the market participants with heterogeneous patience levels, influences the market thickness and its resulting throughput. This issue is studied  in a rich recent literature on dynamic stochastic matching (see, e.g.,  \citet{anderson2017efficient,akbarpour2020thickness,ozkan2020dynamic}), where the standard assumption is that agents have independent exponentially-distributed patience levels. Closer to our setting, the reward-maximization version of this problem admits constant-factor approximation algorithms based on linear programming relaxations~\citep{collina2020dynamic,aouad2022dynamic}. To our knowledge, none of these models admits efficient approximation schemes, even with further restrictions on the underlying graph \citep{banerjee,kessel2022stationary}. Moreover, although dynamic stochastic matching models on specialized graphs bear certain resemblance with our setting, there are substantial differences in terms of structural assumptions, e.g., finite vs.\ infinite horizon and discrete vs.\ continuous time.
\section{Technical Overview} \label{sec:overview}

In what follows, we provide a high-level overview of our algorithmic ideas and their analysis. Specifically, Section~\ref{subsec:reduce_average} describes our customer classification method, followed by proving that class-ordered policies are sufficiently strong to attain near-optimal expected rewards. Sections~\ref{subsec:round_values} and~\ref{subsec:round_prob} are dedicated to explaining how customer rewards and departure probabilities are rounded, and to establishing a number of auxiliary technical claims in this context. In Section~\ref{subsec:rev_preserve}, we succinctly analyze the effects of these alterations, leaving most finer details to be discussed in Section~\ref{sec:proof_thm_rev_diff_direction}. Finally, Sections~\ref{subsec:alg_consequences} and~\ref{subsec:alg_rounded} discuss the algorithmic consequences of this analysis, leading to a deterministic $O( n^{ O_{ \eps }( \log^2 n ) } )$-time dynamic programming approach for computing a $(1-\eps)$-approximate service policy.

\subsection{Preliminary step: Reduction to average customers} \label{subsec:reduce_average}

We start off by arguing that customers with ``very small'' or ``very large'' departure probabilities can be separately served by employing simple non-adaptive priority rules, ensuring that their expected reward contribution nearly matches the analogous quantity with respect to an optimal adaptive policy. Specifically, since these regimes of departure probabilities turn out to be easy to handle, we proceed by devising a general reduction from arbitrarily structured instances to ones where all customers are associated with ``average'' departure probabilities. The cut-off values that separate between these classes are defined next.
 
\paragraph{Customer classification.} Given an error parameter $\eps \in (0, \frac{ 1 }{ 4 })$, we begin by classifying customers into three types -- stickers, quitters, and average -- depending on the magnitude of their departure probabilities. At least intuitively, stickers depart within a single stage with very low probability, quitters depart with very high probability, and those who do not fall into these two extremes are called average. Specifically, we say that customer $i$ is a sticker when $p_i < \frac{ \eps }{ n^2 }$. On the other hand, customer $i$ is a quitter when $p_i > 1 - \frac{ \eps }{ n }$. Finally, customers with $p_i \in [\frac{ \eps }{ n^2 }, 1 - \frac{ \eps }{ n }]$ will be referred to as average. We denote the corresponding subsets of customers as ${\cal C}_{\mystick}$, ${\cal C}_{\myquit}$, and ${\cal C}_{\myavg}$, respectively.
 
\paragraph{Class-ordered policies are near optimal.} We say that an adaptive policy ${\cal S} : \bbN \times 2^{[n]} \to [n]$ is class-ordered when it satisfies the next two conditions:
\begin{enumerate}
    \item Quitters can only be served at stage 1.

    \item Stickers are served only in stages $n+1,\ldots,2n$, by picking the highest-reward sticker who is still available in each stage.
\end{enumerate}
The reason for choosing the notion of ``class ordered'' to describe such policies is that we first serve quitters, then average customers, and finally stickers. Interestingly, by property~1, any class-ordered policy can serve at most one quitter. Moreover, property~2 requires us to consider an extended definition of adaptive service policies where, to enable serving customers in stages $n+1,\ldots,2n$, we allow policies not to serve any customer at any given stage. Notation-wise, we make use of ${\cal S}(t,{\cal A}_t) = \perp$ to indicate that the policy ${\cal S}$ does not serve any customer in state $(t,{\cal A}_t)$. 

With these definitions, the next claim shows that class-ordered policies are capable of attaining near-optimal rewards. The proof of this result appears in Appendix~\ref{app:proof_lem_lazy_exists}, noting that we make use of ${\cal S}^*$ to denote a fixed optimal service policy.

\begin{lemma} \label{lem:lazy_exists}
There exists a class-ordered policy ${\cal S}^{ \myco }$ with $\rev( {\cal S}^{ \myco } ) \geq (1 - 6\eps) \cdot \rev( {\cal S}^* )$.
\end{lemma}

As an aside, it is not difficult to verify that such extended policies, where customers may not be served in certain stages, can easily be converted to our original notion of service policies, without any loss in their expected reward. To this end, whenever our policy reaches a state $(t,{\cal A}_t)$ in which ${\cal S}(t,{\cal A}_t) = \perp$, we can simulate its departures $\hat{\cal D}_t \sim {\cal D}_t$ and choose a random action according to the distribution of the next state, ${\cal S}(t+1,{\cal A}_t\setminus \hat{\cal D}_t)$. In other words, our policy treats the system state as being reset to $(t+1,{\cal A}_t\setminus \hat{\cal D}_t)$. This procedure can be repeated for any such state, without loss in the expected reward.

\paragraph{Reduction: Ending up with only average customers?} In what follows, we devise a polynomial-time reduction from our original setting to one comprised of only average customers. To this end, for any subset of average customers $A \subseteq {\cal C}_{\myavg}$, let ${\cal I}^A$ be a newly defined instance, whose initial set of available customers is precisely $A$; all other model ingredients remain unchanged. Now, let us assume that, given any subset $A \subseteq {\cal C}_{\myavg}$, we can compute a $(1-\eps)$-approximate service policy ${\cal S}^{\approx A}$ with respect to ${\cal I}^A$. In other words, letting ${\cal S}^{*A}$ be an optimal policy for the latter instance, $\rev_{ {\cal I}^A }( {\cal S}^{\approx A} ) \geq (1 - \eps) \cdot \rev_{ {\cal I}^A }( {\cal S}^{*A} )$, with the convention that subscripts of $\rev_{\cdot}$ will indicate the instance being considered.

In this setting, we initially guess the identity of customer $i_1 = {\cal S}^{ \myco }(1,{\cal A}_1)$, namely, the one being served by the class-ordered policy ${\cal S}^{ \myco }$ in stage~1. Knowing who this customer is, the policy ${\cal S}$ we define for our original instance operates as follows:
\begin{itemize}
    \item {\em Stage 1}: Here, we duplicate the service decision made by ${\cal S}^{ \myco }$, meaning that ${\cal S}(1, {\cal A}_1) = i_1$.
    
    \item {\em Stages $2, \ldots, n$}: Let $A \subseteq {\cal C}_{\myavg}$ be the set of available average customers at the beginning of stage $2$. Namely, $A$ is the specific realization of ${\cal A}_2 \cap {\cal C}_\myavg$ we observe at this time, where ${\cal A}_2 = {\cal A}_1 \setminus ( \{ i_1 \} \cup {\cal D}_1 )$. Then, across stages $2, \ldots, n$, our policy simply ignores all non-average customers, and employs the approximate policy ${\cal S}^{\approx A}$.
    
    \item {\em Stages $n+1, \ldots, 2n$:} We revert back to following the class-ordered policy ${\cal S}^{\myco}$. That is, only stickers will be served, picking the highest-reward sticker who is still available at each stage.
\end{itemize}
The next claim, whose proof appears in Appendix~\ref{app:reduction}, shows that the policy we have just proposed is indeed near-optimal. 

\begin{lemma} \label{lem:reduction_average}
$\rev({\cal S})\geq (1-7\eps)\cdot \rev({\cal S}^*)$.
\end{lemma}

\paragraph{Intermediate summary.} Given this result, our reduction can be employed to convert any black-box quasi-PTAS for the restricted class of instances with only average customers into a quasi-PTAS for arbitrarily-structured instances, potentially including customers of all three classes. Hence, going forward, we restrict attention to instances exclusively formed by average customers. For notational convenience, we assume that ${\cal C}_{\myavg} = [n]$, whereas ${\cal C}_{\mystick} ={\cal C}_{\myquit} = \emptyset$. 
 
\subsection{Step 1: Rounding customer rewards} \label{subsec:round_values}

In what follows, we explain how to round the collection of customer rewards $r_1, \ldots, r_n$, ending up with only $O( \frac{ 1 }{ \eps } \log \frac{ n }{ \eps } )$ distinct values while losing a negligible fraction of the optimal expected reward. It is worth pointing out that analyzing the effects of this rounding procedure on our expected reward will be rather straightforward. In essence, the overall customer departure process remains probabilistically identical; the only difference would be that, whenever a customer is served, we collect a slightly smaller reward.

\paragraph{The rounded rewards.} To this end, let $r_{\max} = \max_{i \in [n]} r_i$ be the largest reward of any customer. For each customer $i \in [n]$, we define a rounded-down reward $\tilde{r}_i$ depending on how his/her original reward $r_i$ relates to $r_{\max}$:
\begin{itemize}
\item For customers with $r_i \in [\frac{ \eps }{ n } \cdot r_{\max}, r_{\max}]$, we set $\tilde{r}_i$ as the result of rounding $r_i$ down to the nearest power of $1 + \eps$.

\item For the remaining customers, with $r_i \leq \frac{ \eps }{ n } \cdot r_{\max}$, their rounded-down reward is set as $\tilde{r}_i = 0$.
\end{itemize}
One can easily verify that this rounding procedure indeed creates $O( \frac{ 1 }{ \eps } \log \frac{ n }{ \eps } )$ distinct values.

\paragraph{Reward effects.} To account for the resulting loss in expected reward, let $\rev^{ (r) }( \cdot )$ represent our original expected reward function, defined with $\{ r_i \}_{i \in [n]}$ as customer rewards, and let $\rev^{ (\tilde{r}) }( \cdot )$ be the analogous  function, defined with respect to $\{ \tilde{r}_i \}_{i \in [n]}$. Lemma~\ref{lem:effect_rnd_v} below, whose proof is provided in Appendix~\ref{app:proof_lem_effect_rnd_v}, compares the expected rewards $\rev^{ (\tilde{r}) }( {\cal S} )$ and $\rev^{ (r) }( {\cal S} )$ of every adaptive service policy ${\cal S}$. Below, ${\cal S}^*$ designates a fixed optimal service policy with respect to the original expected reward function $\rev^{ (r) }( \cdot )$. 

\begin{lemma} \label{lem:effect_rnd_v}
For every adaptive service policy ${\cal S}$,
\[ (1 - \eps) \cdot \rev^{ (r) }( {\cal S} ) - \eps \cdot \rev^{ (r) }( {\cal S}^* ) ~~\leq~~ \rev^{ (\tilde{r}) }( {\cal S} ) ~~\leq~~ \rev^{ (r) }( {\cal S} ) \ . \]
\end{lemma}

\paragraph{Intermediate summary.} Going forward, we utilize the rounded customer rewards $\{\tilde{r}\}_{i\in [n]}$ in place of the original ones, $\{r_i\}_{i\in [n]}$. In light of Lemma~\ref{lem:effect_rnd_v}, any $(1-\eps)$-approximate service policy with respect to this modified instance yields a $(1-3\eps)$-approximation for the optimal expected reward in the original instance. For simplicity, we will continue to denote the customer rewards by $\{r_i\}_{i\in [n]}$, despite utilizing their rounded counterparts $\{\tilde{r}_i\}_{i\in [n]}$. 

\subsection{Step 2: Rounding departure probabilities} \label{subsec:round_prob}

We proceed by describing our rounding procedure for the collection of departure probabilities $p_1, \ldots, p_n$, ending up with only $O( \frac{ 1 }{ \eps^2 } \log \frac{ n }{ \eps } )$ distinct values. In sharp contrast to step~1, one can easily observe that, due to considering rounded-up probabilities, we are in fact altering the customer departure process. Rigorously analyzing the effects of this alteration on our expected reward turns out to be highly non-trivial, and therefore, we  separately study this question in Section~\ref{subsec:rev_preserve}.

\paragraph{The rounded probabilities.} With respect to the departure probability $p_i$ of each customer $i \in [n]$, we define two counterpart probabilities, an up-rounding $p^{\uparrow}_i$ and a down-rounding $p^{\downarrow}_i$; by inspecting their exact definitions below, one can easily verify that $p_i \in [p^{\downarrow}_i, p^{\uparrow}_i]$. For this purpose, we remind the reader that the reduction outlined in Section~\ref{subsec:reduce_average} allows us to focus on problem instances consisting of only average customers, meaning that $p_i \in [\frac{ \eps }{ n^2 }, 1 - \frac{ \eps }{ n }]$ for all $i\in [n]$. Consequently, letting $\delta = \frac{\eps^2}{16}$, the rounded probabilities $p^{\uparrow}_i$ and $p^{\downarrow}_i$ are determined based on the following case disjunction:
\begin{itemize}
    \item When $p_i \in [\frac{ \eps }{ n^2 }, \frac{ \eps }{ 4 }]$, we set $p^{\uparrow}_i = U( p_i )$ and $p^{\downarrow}_i = \frac{ U( p_i ) }{ 1 + \delta }$, where $U(\cdot)$ is an operator that rounds its argument up to the nearest power of $1 + \delta$.

        \item When $p_i \in (\frac{ \eps }{ 4 }, 1 - \frac{ \eps }{ n }]$, we set $p^{\uparrow}_i = 1 - D(1 - p_i)$ and $p^{\downarrow}_i = 1 - (1 + \delta) \cdot D(1 - p_i)$. In this case, the operator $D(\cdot)$ rounds its argument down to the nearest power of $1 + \delta$.
\end{itemize}
It is worth pointing out that we will make extensive algorithmic use of rounded-up probabilities, whereas the rounded-down ones are introduced solely for purposes of analysis and their intended role will become clear in subsequent sections.

\paragraph{Property 1: Number of different values.} The first important observation is that, although the original set of departure probabilities $\{ p_i \}_{i \in [n]}$ may consist of $n$ distinct values, their rounded counterparts $\{ p^{ \uparrow }_i \}_{i \in [n]}$ and $\{ p^{ \downarrow }_i \}_{i \in [n]}$ only take $O( \frac{ 1 }{ \eps^2 } \log \frac{ n }{ \eps } )$ distinct values. Recalling that $\delta = \frac{ \eps^2 }{ 16 }$, this claim can be verified by noting that when $p_i \in [\frac{ \eps }{ n^2 }, \frac{ \eps }{ 4 }]$, there are only $O( \frac{ 1 }{ \delta } \log n )$ powers of $1 + \delta$ to which $U(\cdot)$ can map any probability $p_i$ within this range. Similarly, when $p_i \in (\frac{ \eps }{ 4 }, 1 - \frac{ \eps }{ n }]$, there are only $O( \frac{ 1 }{ \delta } \log \frac{ n }{ \eps } )$ powers of $1 + \delta$ to which $D(\cdot)$ can map $1 - p_i$.

\paragraph{Property 2: Departures over multiple stages.} Taking the viewpoint of a single customer, we proceed by establishing a particularly useful property that will be utilized by our subsequent analysis. As explained below, when the appropriate corrections are made, the probability of any average customer to remain in the system over any sequence of successive stages when departures are governed by $\{ p^{ \uparrow }_i \}_{i \in [n]}$ can be related to the analogous probability with respect to $\{ p^{ \downarrow }_i \}_{i \in [n]}$. To formalize this notion, we consider two scenarios, one where the latter sequence is short and the other where it is of arbitrary length.

Let us call a sequence of $\Delta$ successive stages ``short'' when $\Delta < \frac{ 1 }{ \eps }$. The first property we establish shows that the probability $( 1 - p^{ \uparrow }_i )^{ \Delta }$ of an average customer $i$ not to depart along such a sequence with respect to the rounded-up departure probability, $p^{ \uparrow }_i$, nearly matches the analogous probability $( 1 - p^{ \downarrow }_i )^{ \Delta }$ with respect to the rounded-down  probability, $p^{ \downarrow }_i$. The proof of this claim appears in Appendix~\ref{app:proof_lem_comp_prob_short}.

\begin{lemma} \label{lem:comp_prob_short}
For every customer $i \in {\cal C}_{\myavg}$ and for every $\Delta \in [0, \frac{ 1 }{ \eps })$,
\[ ( 1 - p^{ \uparrow }_i )^{ \Delta } ~~\geq~~ (1 - \eps) \cdot ( 1 - p^{ \downarrow }_i )^{ \Delta } \ . \]
\end{lemma}

Now, when the sequence of stages is of arbitrary length, elementary examples demonstrate that the probabilities $( 1 - p^{ \uparrow }_i )^{ \Delta }$ and $( 1 - p^{ \downarrow }_i )^{ \Delta }$ could be exponentially far apart. For instance, when $p_i = \frac{ 1 }{ 2 }$, our rounding approach could set $p^{ \uparrow }_i = 1 - \frac{ 1 }{ 1 + \delta } \cdot \frac{ 1 }{ 2 }$ and $p^{ \downarrow }_i = \frac{ 1 }{ 2 }$, in which case $\frac{ ( 1 - p^{ \downarrow }_i )^n }{ ( 1 - p^{ \uparrow }_i )^n } = \exp ( \Omega( \delta n ) ) = \exp ( \Omega( \eps^2 n ) )$. Motivated by this observation, the next property argues that these probabilities are still comparable when we are cutting an $\eps$-fraction of the $p^{ \uparrow }$-related sequence. The proof of this claim appears in Appendix~\ref{app:proof_lem_comp_prob_lengthy}.

\begin{lemma} \label{lem:comp_prob_lengthy}
For every customer $i \in {\cal C}_{\myavg}$ and for every $\Delta \geq 0$,
\[ ( 1 - p^{ \uparrow }_i )^{(1 - \eps) \Delta } ~~\geq~~ ( 1 - p^{ \downarrow }_i )^{ \Delta } \ . \]
\end{lemma}

\subsection{Reward near-preservation bounds} \label{subsec:rev_preserve}

\paragraph{The rounded-up process.} Now let us consider a natural alteration of the random process described in Section~\ref{subsec:model_desc}, where the original departure probabilities $\{ p_i \}_{i \in [n]}$ are substituted by their rounded-up counterparts $\{ p^{ \uparrow }_i \}_{i \in [n]}$; all other model ingredients remain unchanged. In this context, the analytical question we examine is motivated by the intuitive concern that one may be losing much of the expected reward, due to customers who are now departing at a faster rate. Put in concrete terms, letting $\rev^{ (p) }( \cdot )$ represent our original expected reward function, defined with $\{ p_i \}_{i \in [n]}$ as departure probabilities, we use $\rev^{ (p^{ \uparrow }) }( \cdot )$ to denote its analogous reward function, defined with respect to $\{ p^{ \uparrow }_i \}_{i \in [n]}$. Then, the specific questions we wish to address, whose algorithmic implications will be examined in Section~\ref{subsec:alg_consequences}, can be succinctly stated as follows:
\begin{quote}
\begin{enumerate}
\item {\em Given an optimal policy ${\cal S}^*$ for the original process, is there a corresponding policy ${\cal S}^{ \uparrow }$ for the rounded-up process with $\rev^{ (p^{ \uparrow }) }( {\cal S}^{ \uparrow } ) \approx \rev^{ (p) }( {\cal S}^* )$?} 

\item {\em Conversely, given a policy ${\cal S}$ for the rounded-up process, is there a corresponding policy ${\cal S}^{ \downarrow }$ for the original process with $\rev^{ (p) }( {\cal S}^{ \downarrow } ) \approx \rev^{ (p^{ \uparrow }) }( {\cal S} )$?}
\end{enumerate}
\end{quote}

\paragraph{Answer 1: Negligible reward loss in rounding up.} Our main technical result answers question~1 in the affirmative by showing that, despite considering customers with higher departure rates, we can nearly match the expected reward $ \rev^{ (p) }( {\cal S}^* )$ in the rounded-up process using a carefully constructed policy ${\cal S}^{ \uparrow }$. Specifically, we dedicate Section~\ref{sec:proof_thm_rev_diff_direction} to proving the next result.

\begin{theorem} \label{thm:rev_diff_direction}
There exists an adaptive service policy ${\cal S}^{ \uparrow } : \bbN \times 2^{[n]} \to [n]$ satisfying
\[ \rev^{ (p^{ \uparrow }) }( {\cal S}^{ \uparrow } ) ~~\geq~~ (1 - 2\eps) \cdot \rev^{ (p) }( {\cal S}^* ) \ . \]
\end{theorem}

At least intuitively, the main technical hurdle in establishing this result lies in the cumulative effects of our rounding errors on customers' departure probabilities. As can be concluded from Lemma~\ref{lem:comp_prob_short}, along the first few stages, the probability that any given customer remains available in the rounded-up process does not deviate by much from the analogous probability with respect to the original process. However, as we progress to later stages, small per-stage rounding errors accumulate over time, creating an exponential gap between these probabilities, since customers are now departing at a higher rate. Therefore, it is  unclear whether there exists a policy ${\cal S}^{ \uparrow }$ in the rounded-up process that is capable of reward-wise competing against the optimal policy ${\cal S}^*$ in the original process. As further explained in Section~\ref{sec:proof_thm_rev_diff_direction}, the proof of Theorem~\ref{thm:rev_diff_direction} is based on the construction of a probabilistic coupling, which quite surprisingly, enables us to ``realign'' the two departure processes and concede only a small loss of reward in expectation.

\paragraph{Answer 2: No reward loss in restoring.} Now, in the opposite direction, it turns out that when moving from the rounded-up process back to the original one, we can actually preserve the expected reward of any service policy. At least intuitively, this property proceeds by observing that, when departure probabilities are decreased, customers are more likely to remain available up to any given stage, which may only increase the expected reward of any given policy, assuming it is suitably adapted. We formalize this intuition in the next theorem, whose constructive proof is provided in Appendix~\ref{app:proof_thm_rev_diff_easy}.

\begin{theorem} \label{thm:rev_diff_easy}
Let $\{ p_i^- \}_{i \in [n]}$ and $\{ p_i^+ \}_{i \in [n]}$ be two collections of departure probabilities, with $p_i^- \leq p_i^+$ for all $i \in [n]$. Then, for any adaptive policy ${\cal S} : \bbN \times 2^{[n]} \to [n]$, there exists a policy ${\cal S}^{ \downarrow } : \bbN \times 2^{[n]} \to [n]$ satisfying $\rev^{ (p^-) }( {\cal S}^{ \downarrow } ) \geq \rev^{ (p^+) }( {\cal S} )$.
\end{theorem}

\subsection{Algorithmic consequences} \label{subsec:alg_consequences}

Along the transformations described in Sections~\ref{subsec:round_values} and~\ref{subsec:round_prob}, we have shown that our resulting rounded-up process is guaranteed to satisfy two structural properties. First, the customer rewards $\{ r_i \}_{i \in [n]}$ take only $O( \frac{ 1 }{ \eps } \log \frac{ n }{ \eps } )$ distinct values, and second, the departure probabilities $\{ p^{ \uparrow }_i \}_{i \in [n]}$ take only $O( \frac{ 1 }{ \eps^2 } \log \frac{ n }{ \eps } )$ distinct values. Based on these characteristics, we argue that an optimal adaptive policy in this setting can be computed in quasi-polynomial time by means of dynamic programming. This result is formally stated in Theorem~\ref{thm:solve_rounded}, whose proof appears in Section~\ref{subsec:alg_rounded}. In a nutshell, one is no longer required to keep track of the precise identity of each available customer, as in the general-purpose state description $(t, {\cal A}_t)$ of Section~\ref{subsec:model_desc}. Instead, we have created an ideal setting, where customers are segmented into $O_{ \eps }( \log^2 n )$ reward-probability classes, and due to their identical probabilistic role, it suffices to know the number of customer available out of each such class. 

\begin{theorem} \label{thm:solve_rounded}
An optimal adaptive policy for the rounded-up process can be computed in $O( n^{ O_{ \eps }( \log^2 n ) } )$ time.
\end{theorem}

Motivated by this result, our algorithmic approach is rather straightforward. We first create all input ingredients related to the rounded-up process, and subsequently employ Theorem~\ref{thm:solve_rounded} to obtain an optimal policy ${\cal S}^{ \uparrow* }$ for the resulting instance. We then utilize Theorem~\ref{thm:rev_diff_easy}, producing an adaptive service policy $\tilde{\cal S}$ with respect to the original process, such that $\rev^{ (p) }( \tilde{\cal S} ) \geq \rev^{ (p^{ \uparrow }) }( {\cal S}^{ \uparrow* } )$, noting that the original departure probabilities $\{ p_i \}_{i \in [n]}$ are dominated by their rounded-up counterparts $\{ p^{ \uparrow }_i \}_{i \in [n]}$. We conclude that $\tilde{\cal S}$ is in fact a near-optimal policy with respect to the original instance, since
\begin{eqnarray*}
\rev^{ (p) }( \tilde{\cal S} ) & \geq & \rev^{ (p^{ \uparrow }) }( {\cal S}^{ \uparrow* } ) \\
& \geq & \rev^{ (p^{ \uparrow }) }( {\cal S}^{ \uparrow } ) \\
& \geq & (1 - 2\eps) \cdot \rev^{ (p) }( {\cal S}^* ) \ .
\end{eqnarray*}
Here, the second inequality is implied by the optimality of ${\cal S}^{ \uparrow* }$, meaning in particular that its expected reward (in $\rev^{ (p^{ \uparrow }) }$-terms) is at least as large as that of the policy ${\cal S}^{ \uparrow }$, whose existence has been established in Theorem~\ref{thm:rev_diff_direction}. The third inequality is precisely the one stated in the latter theorem.

\subsection{Proof of Theorem~\ref{thm:solve_rounded}} \label{subsec:alg_rounded}

Let $\{ r_{(\psi)} \}_{\psi \in \Psi}$ be the collection of distinct values taken by the customer rewards $r_1, \ldots, r_n$. Similarly, let $\{ p_{ (\lambda) } \}_{ \lambda \in \Lambda }$ be the set of distinct values taken by the departure probabilities  $p^{ \uparrow }_1, \ldots,  p^{ \uparrow }_n$. In what follows, we argue that one can compute an optimal service policy in $O( n^{ O( |\Psi| \cdot |\Lambda| ) } )$ time by means of dynamic programming. According to the discussion in Section~\ref{subsec:alg_consequences}, we know that $|\Psi| = O( \frac{ 1 }{ \eps } \log \frac{ n }{ \eps } )$ and $|\Lambda| = O( \frac{ 1 }{ \eps^2 } \log \frac{ n }{ \eps } )$, immediately leading to the $O( n^{ O_{ \eps }( \log^2 n ) } )$ running time stated in Theorem~\ref{thm:solve_rounded}.

\paragraph{State description.} For this purpose, for every $\psi \in \Psi$ and $\lambda \in \Lambda$, let us make use of ${\cal C}_{\psi, \lambda}$ to denote the collection of customers with reward $r_{(\psi)}$ and departure probability $p_{ (\lambda) }$. The important observation is that, within any such reward-probability class, all customers play precisely the same probabilistic role in our model. Consequently, in order to compute an optimal service policy, the exact identity of each available customer within any given reward-probability class is an overly refined state description, and instead, it suffices to keep track only of the combined number of such customers. Motivated by this observation, each state $(t, {\cal N}_t)$ of our dynamic program consists of the next two parameters:
\begin{itemize}
    \item The current stage index, $t$, taking one of the values $1, \ldots, n+1$.
    
    \item A $(|\Psi| \cdot |\Lambda|)$-dimensional count vector, ${\cal N}_t$, with the convention that ${\cal N}_{t, (\psi, \lambda)}$ stands for the number of available customers at the beginning of stage $t$ out of the reward-probability class ${\cal C}_{\psi, \lambda}$. Clearly, there are only $\prod_{ \psi \in \Psi, \lambda \in \Lambda } (|{\cal C}_{\psi, \lambda}| + 1) = O( n^{ O( |\Psi| \cdot |\Lambda| ) } )$ count vectors that are relevant to our purposes.
\end{itemize}

\paragraph{Value function and recursive equations.} For every state $(t, {\cal N}_t)$, our value function $\rev(t, {\cal N}_t)$ specifies the maximum-possible expected cumulative reward that can be attained along the process in question, starting at state $(t, {\cal N}_t)$. To express this function in recursive form, we first define terminal states as those where there are no available customers (i.e., ${\cal N}_t = \vec{0}$), in which case one clearly has $\rev(t, \vec{0}) = 0$. Now, for general states $(t, {\cal N}_t)$ with ${\cal N}_t \neq \vec{0}$, suppose that the optimal policy with respect to this state decides to serve a customer belonging to class ${\cal C}_{\psi^*, \lambda^*}$. Then, our immediate reward is $r_{(\psi^*)}$, and we proceed to stage $t+1$ with the random count vector ${\cal N}_t - e_{ \psi^*, \lambda^*} - {\cal D}_{ {\cal N}_t - e_{ \psi^*, \lambda^*} }$. To parse the latter expression, we note that:
\begin{itemize}
    \item The vector $e_{ \psi^*, \lambda^*}$ is simply a standard unit vector, by which we decrement ${\cal N}_t$ due to the customer who has just been served.
    
    \item The vector ${\cal D}_{ {\cal N}_t - e_{ \psi^*, \lambda^*} }$ is random, representing the set of departing customers at stage $t$. More precisely, given that the departure probabilities in each class ${\cal C}_{\psi, \lambda}$ are uniformly equal to $p_{ (\lambda) }$, and since customer departures are independent, it follows that $({\cal D}_{ {\cal N}_t - e_{ \psi^*, \lambda^*} })_{(\psi, \lambda)} \sim \mathrm{Binomial}(({\cal N}_t - e_{ \psi^*, \lambda^*})_{(\psi, \lambda)}, p_{ (\lambda) } )$; moreover, $\{ {\cal D}_{ {\cal N}_t - e_{ \psi^*, \lambda^*} } \}_{(\psi, \lambda)}$ are mutually independent. 
\end{itemize}
Based on this discussion, we infer that the optimal policy consists in picking, out of the still-active reward-probability classes ${\cal C}_{\psi, \lambda}$, one that maximizes $r_{(\psi)} + \expar{\rev(t+1, {\cal N}_t - e_{ \psi, \lambda} - {\cal D}_{ {\cal N}_t - e_{ \psi, \lambda} } ) }$, where the expectation here is taken over the randomness in ${\cal D}_{ {\cal N}_t - e_{ \psi, \lambda} }$. In other words, 
\[ \rev(t, {\cal N}_t) ~~=~~ \max_{ \MyAbove{ \psi \in \Psi, \lambda \in \Lambda: }{{\cal N}_{t, (\psi, \lambda)} \geq 1 } } \left\{   r_{(\psi)} + \ex{\rev(t+1, {\cal N}_t - e_{ \psi, \lambda} - {\cal D}_{ {\cal N}_t - e_{ \psi, \lambda} } ) } \right\} \ . \]

\paragraph{Running time.} From a running time perspective, a straightforward implementation of our dynamic program operates in $O( n^{ O( |\Psi| \cdot |\Lambda| ) } )$ time. Indeed, the state space over which we compute the function $\rev$ of size $O( n^{ O( |\Psi| \cdot |\Lambda| ) } )$. In addition, based on the recursive equations above, the main bottleneck in evaluating each such state is that of computing its inner expectation, $\expar{\rev(t+1, {\cal N}_t - e_{ \psi, \lambda} - {\cal D}_{ {\cal N}_t - e_{ \psi, \lambda} } ) }$. However, one can implement this procedure in $O( n^{ O( |\Psi| \cdot |\Lambda| ) } )$ time by enumerating over the support of ${\cal D}_{ {\cal N}_t - e_{ \psi, \lambda} }$, namely,
\[ \ex{\rev(t+1, {\cal N}_t - e_{ \psi, \lambda} - {\cal D}_{ {\cal N}_t - e_{ \psi, \lambda} } ) } ~~=~~ \sum_{D \leq {\cal N}_t - e_{ \psi, \lambda}} \pr{ {\cal D}_{ {\cal N}_t - e_{ \psi, \lambda} } = D } \cdot \rev(t+1, {\cal N}_t - e_{ \psi, \lambda} - D ) \ . \]
We mention in passing that probabilities of the form appearing above, $\prpar{ {\cal D}_{ {\cal N}_t - e_{ \psi, \lambda} } = D }$, can easily be derived in $n^{ O(1) }$ time, since each coordinate of ${\cal D}_{ {\cal N}_t - e_{ \psi, \lambda} }$ follows a known Binomial distribution (trivially upper-bounded by $n$), and since these coordinates are mutually independent.
\section{Reward Preservation of the Rounded-Up Process} \label{sec:proof_thm_rev_diff_direction}

This section is dedicated to presenting a coupling-based proof of Theorem~\ref{thm:rev_diff_direction}, showing that there exists an adaptive service policy ${\cal S}^{ \uparrow }$ for the rounded-up process whose expected reward nearly matches that of the optimal policy ${\cal S}^*$ for our original process. For ease of exposition, we first describe the overall structure of our proof.

\subsection{Proof outline} \label{subsec:proof_outline}

As it turns out, direct comparisons between expected rewards in the original and rounded-up processes are rather convoluted. Hence, following the discussion in Section~\ref{subsec:rev_preserve}, our analysis will utilize an intermediate alteration of the original process, in which the departure probabilities $\{ p_i \}_{i \in [n]}$ are substituted by their rounded-down counterparts $\{ p^{ \downarrow }_i \}_{i \in [n]}$; all other model ingredients remain unchanged. In this context, we make use of $\rev^{ (p^{ \downarrow }) }( \cdot )$ to denote the resulting expected reward function. Now, by Theorem~\ref{thm:rev_diff_easy}, we observe that with respect to the optimal policy ${\cal S}^{ * }$ for the original process, there exists a corresponding policy ${\cal S}^{ *\downarrow }$ for the rounded-down process whose expected reward in $\rev^{ (p^{ \downarrow }) }$-terms dominates that of ${\cal S}^{ * }$ in $\rev^{ (p) }$-terms. In other words, $\rev^{ (p^{ \downarrow}) }( {\cal S}^{ * \downarrow } ) \geq \rev^{ (p) }( {\cal S}^{ * } )$.  For this reason, in order to prove Theorem~\ref{thm:rev_diff_direction}, it suffices to argue that the rounded-up process admits a policy ${\cal S}^{ \uparrow }$ satisfying
$\rev^{ (p^{ \uparrow }) }( {\cal S}^{ \uparrow } ) \geq (1 - 2\eps) \cdot \rev^{ (p^{ \downarrow}) }( {\cal S}^{ *\downarrow } )$.

In the remainder of this section, we prove the existence of such a policy ${\cal S}^{\uparrow}$, operating in the rounded-up process. Informally speaking, our specific design of the policy ${\cal S}^{\uparrow}$ ensures that we are approximately imitating ${\cal S}^{ *\downarrow }$, while making-up for the rounding errors between $p^{\uparrow}$ and $p^{\downarrow}$ by skipping carefully-selected stages. The main technical ingredient employed to define the policy ${\cal S}^{\uparrow}$ and to lower-bound its expected reward is a probabilistic coupling between the rounded-down and rounded-up processes, which is presented in Sections~\ref{subsec:prelim2}-\ref{subsec:coupling_prop}. Based on this coupling, in Section~\ref{subsec:up_policy}, we specify how our policy ${\cal S}^{\uparrow}$ operates, essentially duplicating the service decisions made by ${\cal S}^{ *\downarrow }$ with respect to simulated system states. By exploiting this construction, in Section~\ref{subsec:analysis}, we argue that $\rev^{ (p^{ \uparrow }) }( {\cal S}^{ \uparrow } ) \geq (1 - 2\eps) \cdot \rev^{ (p^{ \downarrow}) }( {\cal S}^{ *\downarrow } )$, thereby concluding the proof of Theorem~\ref{thm:rev_diff_direction}.

\subsection{Preliminary definitions and notation} \label{subsec:prelim2}

In the upcoming discussion, we introduce  a number of auxiliary definitions that will serve as key ingredients of our coupling construction. We further describe an alternative way of viewing how customers make their departure decisions within the rounded-down and rounded-up processes; these are primarily meant to simplify certain parts of our analysis. 

\paragraph{Milestones and stage alignment.} To begin, assuming without loss of generality that $\frac{ 1 }{ \eps }$ takes an integer value, let $\gamma \in \{ 1, \ldots, \frac{1}{\eps} \}$ be a so-called shifting parameter, whose precise choice will be specified in Section~\ref{subsec:analysis}. As a side note, all arguments up until then work for any possible value of $\gamma$ within $\{ 1, \ldots, \frac{1}{\eps} \}$. Given this parameter, we define a special set of stages, ${\cal M}_{\gamma} = \{ t_0, t_1, \ldots \}$, referred to as milestones; every other stage is said to be regular. Here, the $0$-th milestone corresponds to $t_0 = 0$, which is clearly not a concrete stage but rather a convenient notation for subsequent definitions. Then, for every $k \in \bbN$, the $k$-th milestone corresponds to stage $t_k = \frac{k-1}{\eps} + \gamma$. In addition, we define the mapping $\mu: \bbN \rightarrow \bbN \setminus {\cal M}_{\gamma}$ that assigns each stage $t\in \bbN$ to the regular stage $\mu(t) = t+k$, where $k$ is the unique integer for which $t + k \in (t_k,t_{k+1})$. Put differently, $\mu(t)$ is simply the $t$-th regular stage. Since $\mu$ is clearly bijective, each regular stage $\tau\in \bbN \setminus {\cal M}_{\gamma}$ can be reciprocally mapped to a unique stage $t = \mu^{-1}(\tau)$ for which $\mu(t) = \tau$.

\paragraph{The $\bs{({{\cal X}^{\downarrow}},{\cal X}^{\uparrow})}$-representation.} We proceed by developing succinct representations of the rounded-down and rounded-up processes using Bernoulli random variables. Specifically, to capture the rounded-down process, we define a collection ${{\cal X}^{\downarrow}} = \{X^{\downarrow}_{i,t}\}_{i,t}$ of mutually independent Bernoulli random variables, where $\prpar{X^{\downarrow}_{i,t}=1} = p_i^{\downarrow}$ for every customer $i\in [n]$ and every stage $t \geq 1$. The outcome $X^{\downarrow}_{i,t}\in \{0,1\}$ will be interpreted as an indication of whether customer $i$ departs at the end of stage $t$ in the rounded-down process, if s/he is still available then. When this customer has already left the system before stage $t$, the information provided by $X^{\downarrow}_{i,t}$ will be ignored.

With this definition, the state variables $\{({\cal A}_t, {\cal D}_t)\}_{t\geq 1}$ introduced in Section~\ref{subsec:model_desc} can be recursively expressed as a function of ${{\cal X}^{\downarrow}}$ and the sequential service actions taken by the policy in question. For example, focusing on the policy ${\cal S}^{ *\downarrow }$, we denote by $(t,{\cal A}^{{\cal X}^{\downarrow}}_t)$  the system state at the beginning of stage $t$, with ${\cal A}^{{\cal X}^{\downarrow}}_t$ being the subset of remaining customers. As such, the next customer to be served is ${\cal S}^{ *\downarrow }(t,{\cal A}^{{\cal X}^{\downarrow}}_t) \in {\cal A}^{{\cal X}^{\downarrow}}_t$, and each remaining customer $i \in {\cal A}^{{\cal X}^{\downarrow}}_t \setminus \{ {\cal S}^{ *\downarrow }({t,{\cal A}^{{\cal X}^{\downarrow}}_t}) \}$ decides whether to depart from the system or to keep waiting for his/her turn to be served according to the Bernoulli outcome $X_{i,t}^{\downarrow}$. In other words, the random subset of departing customers is ${\cal D}^{{\cal X}^{\downarrow}}_t = \{ i\in {\cal A}^{{\cal X}^{\downarrow}}_t\setminus \{ {\cal S}^{ *\downarrow }({t,{\cal A}^{{\cal X}^{\downarrow}}_t})\}: X_{i,t}^{\downarrow} = 1\}$. When stage $t+1$ begins, we are left with the set of customers ${\cal A}^{{\cal X}^{\downarrow}}_{t+1} = {\cal A}^{{\cal X}^{\downarrow}}_{t} \setminus ({\cal D}^{{\cal X}^{\downarrow}}_t \cup \{{\cal S}^{ *\downarrow }({t,{\cal A}^{{\cal X}^{\downarrow}}_t})\})$.

Similarly, we develop an analogous representation of the rounded-up process. For this purpose, we define ${\cal X}^{\uparrow}  = \{X^{\uparrow}_{i,t}\}_{i,t}$ as a collection of mutually independent Bernoulli random variables, where $\prpar{X^{\uparrow}  _{i,t}=1} =p^{\uparrow}_i$ for every customer $i\in [n]$ and every stage $t \geq 1$. Here, $X^{\uparrow}_{i,t}\in \{0,1\}$ indicates whether customer $i$ departs at the end of stage $t$ in the rounded-up process, if  is still available then. For any given policy, we will denote by $(t,{\cal A}^{{\cal X}^{\uparrow}}_t)$ the system state at the beginning of stage $t$, where ${\cal A}^{{\cal X}^{\uparrow}}_t \subseteq[n]$ stands for the subset of customers remaining at that point in time.

\subsection{Probabilistic coupling} \label{subsec:coupling}

In this section, we construct a probabilistic coupling $({{\cal Y}^{\downarrow}},{\cal Y}^{\uparrow})$ between the unrelated random variables ${{\cal X}^{\downarrow}}$ and ${\cal X}^{\uparrow}$, meaning that ${{\cal Y}^{\downarrow}}\sim {{\cal X}^{\downarrow}}$ and ${{\cal Y}^{\uparrow}}\sim {{\cal X}^{\uparrow}}$. From a policy-design standpoint, we will define ${\cal Y}^{\downarrow}$ and ${\cal Y}^{\uparrow}$ in a correlated way, which will be exploited to develop our new service policy ${\cal S}^{\uparrow}$ in Section~\ref{subsec:up_policy}. Intuitively, the probabilistic relationship between ${{\cal Y}^{\downarrow}}$ and ${\cal Y}^{\uparrow}$ aims to ``align'' the rounded-down and rounded-up processes and to allow for a simpler comparison of their corresponding expected rewards. The natural misalignment arises since each customer $i \in [n]$ leaves at different rates in these two processes, with a departure probability of $p_i^{\downarrow}$ in the rounded-down process and with probability $p^{\uparrow}_i$ in the rounded-up one. At a high level, our coupling compensates for this gap by ``skipping'' milestone stages.

\paragraph{Initial sampling.} We start off by directly defining ${\cal Y}^{\uparrow} = \{Y^{\uparrow}_{i,t}\}_{i,t}$ as a collection of mutually independent Bernoulli random variables with the same distribution as ${{\cal X}^{\uparrow}}$, meaning that $\prpar{Y^{\uparrow}_{i,t} = 1} =p_i^{\uparrow}$ for all $i\in[n]$ and $t\geq1$. Given ${\cal Y}^{\uparrow}$, our method for generating ${\cal Y}^{\downarrow}$ will be based on specifying the conditional probability of ${\cal Y}^{\downarrow}$ relative to the outcomes of ${\cal Y}^{\uparrow}$.  For this purpose, we define several auxiliary collections of Bernoulli random variables -- ${\cal Z}$, ${\cal W}$, and ${\cal V}$ -- which will be utilized to construct ${{\cal Y}^{\downarrow}}$ later on; these are all independent of ${\cal Y}^{\uparrow}$. 
\begin{itemize}
    \item First, we generate a collection of mutually independent Bernoulli random variables ${\cal Z} = \{Z_{i,\tau}\}_{(i,\tau) \in [n] \times (\bbN \setminus {\cal M}_{\gamma})}$, each with a success probability of $\prpar{Z_{i,\tau} = 1}= \frac{ p^{\downarrow}_i }{ p^{\uparrow}_i }$. 
    
    \item Second, independently of ${\cal Z}$, we  generate a collection of mutually independent Bernoulli random variables ${\cal W} = \{W_{i,\tau}\}_{(i,\tau) \in [n] \times {\cal M}_{\gamma}}$, with a success probability of  $\prpar{W_{i,\tau} = 1} = \frac{ p_{i}^{\downarrow} - \xi_{i,\tau} }{ 1-\xi_{i,\tau} }$. Here, $\xi_{i,\tau} = 1- ( \frac{ 1 - p_{i}^{\uparrow} }{ 1-p_{i}^{\downarrow} } )^{{t}_{k} - {t}_{k-1} - 1}$, where $k\geq 1$ is the unique index for which $\tau = t_k$. In this case, it is easy to verify that the latter probability is indeed well-defined, since
    \[ 0~~\leq~~ \xi_{i,\tau} ~~\leq~~ 1 - \left(\frac{1-p_i^{\uparrow}}{1-p_i^{\downarrow}}\right)^{\frac{1}{\eps}-1} ~~\leq~~ 1 - \frac{(1-p_i^{\downarrow})^{\frac{1}{\eps}}}{(1-p_i^{\downarrow})^{\frac{1}{\eps}-1}} ~~=~~  p_{i}^{\downarrow}\ , \]
    where the third inequality follows from Lemma~\ref{lem:comp_prob_lengthy}, instantiated with $\Delta = \frac{1}{\eps}$.
    
    \item Finally, independently of ${\cal Z}$ and ${\cal W}$, we generate a collection of mutually independent Bernoulli random variables ${\cal V} = \{V_{i,\tau}\}_{(i,\tau) \in [n] \times {\cal M}_{\gamma}}$ that follow the same distribution as ${{\cal X}^{\downarrow}}$. Namely, each such variable has a success probability of $\prpar{V_{i,\tau} = 1}= p^{\downarrow}_i$.
\end{itemize}

\paragraph{Inductive construction.} For the purpose of specifying ${\cal Y}^{\downarrow} = \{Y^{\downarrow}_{i,\tau}\}_{i,\tau}$, we consider two distinct cases, depending on whether stage $\tau$ is a milestone or not:
\begin{enumerate}
    \item {\em When stage ${\tau}$ is regular (i.e., ${\tau \notin {\cal M}_{\gamma}}$):} In this case, we define
    \begin{equation} \label{eq:regular}
	Y^{\downarrow}_{i,\tau} ~~=~~ Y^{\uparrow}_{i,\mu^{-1}(\tau)} \cdot Z_{i,\tau} \ .
	\end{equation}
	At least intuitively, $Y^{\downarrow}_{i,\tau}$ is defined in a way that enforces the following property: When customer $i$ departs in stage $\mu^{-1}(\tau)$ of the rounded-up process (i.e., $Y^{\uparrow}_{i, \mu^{-1}(\tau)} = 1$), s/he departs in stage $\tau$ of the rounded-up process with conditional probability $\prpar{ Z_{i,\tau} = 1 } = \frac{ p^{\downarrow}_i }{p^{\uparrow}_i}$, which is exactly the correct ratio of departure rates between these two processes. Conversely, when customer $i$ does not depart in stage $\mu^{-1}(\tau)$ of the rounded-up process, s/he does not depart in stage $\tau$ of the rounded-down process as well.
	
	\item {\em When stage ${\tau}$ is a milestone (i.e., ${\tau \in {\cal M}_{\gamma}}$):} In this case, $\tau = t_k$ for some $k\geq 1$, and we set 
	\begin{eqnarray} \label{eq:x-iu}
	Y^{\downarrow}_{i,\tau} ~~=~~
	\begin{dcases}
	1 - \left(1-W_{i,\tau}\right)\cdot \prod_{s= {t}_{k-1}+1}^{{t}_{k}-1}(1-Y^{\uparrow}_{i,\mu^{-1}(s)}) \qquad & \text{if } Y^{\downarrow}_{i,1} = \cdots = Y^{\downarrow}_{i,\tau-1} = 0  \\
	{V}_{i,\tau} & \text{otherwise}
	\end{dcases}
	\end{eqnarray}
    To gain some preliminary intuition for this definition, one should observe that our specific form of correlation between $Y^{\downarrow}_{i,\tau}$ and $Y^{\uparrow}_{i,\mu^{-1}({t}_{k-1}+1)},\ldots, Y^{\uparrow}_{i,\mu^{-1}({t}_{k}-1)}$ enforces the following property: When customer $i$ departs at one of the stages $\mu^{-1}({t}_{k-1}+1), \ldots, \mu^{-1}({t}_{k}-1)$ of the rounded-up process, s/he departs in the rounded-down process no later than stage $t_{k}$. As such, this definition is meant to ensure that the extra departures incurred by the rounded-up process at stages $\mu^{-1}({t}_{k-1}+1), \ldots, \mu^{-1}({t}_{k}-1)$ are ``compensated'' in the rounded-down process by departures at the current milestone, $t_k$. 
\end{enumerate}

\subsection{Correctness and properties of the coupling construction} \label{subsec:coupling_prop}

We first establish a non-anticipative property relating ${{\cal Y}^{\downarrow}}$ to ${\cal Y}^{\uparrow}$. Specifically, the next claim, whose proof is provided in Appendix~\ref{app:proof_lem_no_future}, shows that our approach for generating the rounded down process does not utilize ``future information'' about the rounded-up process.  To formalize this notion, let us introduce some auxiliary notation to represent the history of any given stochastic process, ${\cal P} = \{P_t\}_{t\in {\bb N}}$. To this end, for every stage $t\geq 1$, the sequence of random variables $P_1, \ldots, P_t$ up to this stage will be denoted by ${\cal P}^{\leq t}$. In complement, ${\cal P}^{> t}$ will stand for the sequence of random variables $P_{t+1}, P_{t+2},\ldots$ from stage $t+1$ onward. 

\begin{lemma} \label{lem:no_future}
There exist deterministic mappings $\nu_1(\cdot)$ and $\nu_2(\cdot)$ such that
\[ {{\cal Y}^{\downarrow}}^{\leq \tau} ~~=~~ 
\begin{dcases}
\nu_1({\cal Z}^{\leq \tau},{\cal W}^{\leq \tau},{\cal V}^{\leq \tau},{{\cal Y}^{\uparrow}}^{\leq  \mu^{-1}(\tau)}) \qquad & \text{\rm if } \tau\notin {\cal M}_\gamma \\
 \nu_2({\cal Z}^{\leq \tau},{\cal W}^{\leq \tau},{\cal V}^{\leq \tau},{{\cal Y}^{\uparrow}}^{\leq  \mu^{-1}(\tau-1)}) \qquad & \text{\rm if } \tau\in {\cal M}_\gamma  
\end{dcases} \]
In particular, ${{\cal Y}^{\downarrow}}^{\leq \tau}$ is independent of $({\cal Z}^{> \tau},{\cal W}^{> \tau},{\cal V}^{> \tau},{{\cal Y}^{\uparrow}}^{>  \mu^{-1}(\tau)})$  for every regular stage $\tau \geq 1$. Similarly, ${{\cal Y}^{\downarrow}}^{\leq \tau}$ is independent of $({\cal Z}^{> \tau},{\cal W}^{> \tau},{\cal V}^{> \tau},{{\cal Y}^{\uparrow}}^{>  \mu^{-1}(\tau-1)})$ for every milestone $\tau \in {\cal M}_{\gamma}$.
\end{lemma}

We proceed to verify that our coupling construction is indeed valid, meaning that the marginal distributions of ${\cal Y}^{\uparrow}$ and ${\cal Y}^{\downarrow}$ are identical to those of ${\cal X}^{\uparrow}$ and ${\cal X}^{\downarrow}$, respectively. As mentioned in Section~\ref{subsec:coupling}, the rounded-up process ${\cal Y}^{\uparrow}$ was directly defined as a collection of mutually independent Bernoulli random variables with the exact same distribution as ${{\cal X}^{\uparrow}}$, implying that the relation ${\cal Y}^{\uparrow} \sim {{\cal X}^{\uparrow}}$ is straightforward. In the remainder of this section, we prove that such a relation also holds in regard to the rounded-down process, arguing that ${{\cal Y}^{\downarrow}}$ and ${{\cal X}^{\downarrow}}$ follow precisely the same distribution.

\begin{lemma}  \label{lem:marginal}
${{\cal Y}^{\downarrow}}\sim {{\cal X}^{\downarrow}}$.
\end{lemma}
\begin{proof}
In order to prove that ${{\cal Y}^{\downarrow}}$ and ${{\cal X}^{\downarrow}}$ are identically distributed, it suffices to show that for every stage $\tau \geq 1$, and for every realization $H$ of the history ${{\cal Y}^{\downarrow}}^{< \tau}$, we have $\prpar{ {Y}^{\downarrow}_{i,\tau} = 1 | {{\cal Y}^{\downarrow}}^{< \tau} = H} = \prpar{ X^{\downarrow}_{i,\tau} = 1 | {{\cal X}^{\downarrow}}^{<\tau}  = H} = p_i^{\downarrow}$. The desired relationship, ${{\cal Y}^{\downarrow}}\sim {{\cal X}^{\downarrow}}$, then follows from the chain rule. Our proof works by induction on $\tau$, where we consider three cases, depending on whether stage $\tau$ is milestone or not, and on whether the history $H$ has at least one departure.

\paragraph{Case~1: Stage $\boldsymbol{\tau}$ is regular.} According to case~1 of our construction, we know that $Y^{\downarrow}_{i,\tau} = Y^{\uparrow}_{i,\mu^{-1}(\tau)} \cdot Z_{i,\tau}$, and therefore 
\begin{eqnarray*}
\pr{\left.{Y}^{\downarrow}_{i,\tau} = 1\right|{{\cal Y}^{\downarrow}}^{<\tau} = H} &= & \pr{\left. \{Y^{\uparrow}_{i,\mu^{-1}(\tau)}= 1\} \wedge \left\{Z_{i,\tau} =1 \right\} \right|{{\cal Y}^{\downarrow}}^{<\tau} = H} \\
&= &\pr{Y^{\uparrow}_{i,\mu^{-1}(\tau)}= 1 } \cdot \pr{Z_{i,\tau} =1}  \\
& = &  p^{\uparrow}_i \cdot 	\frac{p^{\downarrow}_i}{p^{\uparrow}_i} \\
& = &  p^{\downarrow}_i \ ,
\end{eqnarray*}
where the second equality holds since $Z_{i,\tau}$ is independent of $(Y^{\uparrow}_{i,\mu^{-1}(\tau)},{{\cal Y}^{\downarrow}}^{<\tau})$, noting that: (i)~${{\cal Y}^{\downarrow}}^{<\tau} = \nu_1({\cal Z}^{< \tau},{\cal W}^{< \tau},{\cal V}^{< \tau},{{\cal Y}^{\uparrow}}^{<\mu^{-1}(\tau)})$ by Lemma~\ref{lem:no_future}, and (ii)~$Z_{i,\tau}$ is independent of $({\cal Z}^{< \tau},{\cal W},{\cal V},{{\cal Y}^{\uparrow}})$  by construction. Moreover,   Lemma~\ref{lem:no_future} implies that $Y^{\uparrow}_{i,\mu^{-1}(\tau)}$ is independent of ${{\cal Y}^{\downarrow}}^{<\tau}$. 

\paragraph{Case~2: Stage $\bs{\tau}$ is a milestone and $\bs{H \neq \vec{0}}$.} By recalling how case~2 of our  construction works when ${{\cal Y}^{\downarrow}}^{<\tau} \neq \vec{0}$, we have
\begin{eqnarray*}
\pr{\left.{Y}^{\downarrow}_{i,\tau} = 1\right|{{\cal Y}^{\downarrow}}^{<\tau} = H} & = & \pr{ {V}_{i,\tau} = 1 \left|{{\cal Y}^{\downarrow}}^{<\tau} = H \right. }  \\
& = & \pr{{V}_{i,\tau} = 1 } \\
& = & p^{\downarrow}_i \ ,
\end{eqnarray*}
where the  second equality follows from Lemma~\ref{lem:no_future}. 

\paragraph{Case~3: Stage $\bs{\tau}$ is a milestone and $\bs{H = \vec{0}}$.} Let $k$ be the unique index for which $\tau = t_k$. By consulting case~2 of our  construction when ${{\cal Y}^{\downarrow}}^{<\tau} = \vec{0}$, we have
\begin{eqnarray}
&&\pr{\left.{Y}^{\downarrow}_{i,\tau}= 1\right|{{\cal Y}^{\downarrow}}^{<\tau} = H} \nonumber\\
&& \quad  =~~ \pr{\left. \left(1-W_{i,\tau}\right)\cdot \prod_{s = t_{k-1}+1}^{t_k-1}(1-{Y}^{\uparrow}_{i,\mu^{-1}(s)}) = 0 \right|{{\cal Y}^{\downarrow}}^{<\tau} = \vec{0}} \nonumber\\
&& \quad  =~~ \pr{\left. \prod_{s = t_{k-1}+1}^{t_k-1}(1-Y^{\uparrow}_{i,\mu^{-1}(s)}) = 0 \right|{{\cal Y}^{\downarrow}}^{< \tau} = \vec{0} }\nonumber\\
&& \qquad ~~ \mbox{} + \pr{\left. \left\{W_{i,\tau} =1\right\} \wedge \left\{\prod_{s = t_{k-1}+1}^{t_k-1}(1-Y^{\uparrow}_{i,\mu^{-1}(s)}) = 1 \right\} \right|{{\cal Y}^{\downarrow}}^{< \tau} = \vec{0} }\nonumber\\
&& \quad  =~~ \pr{\left. \prod_{s = t_{k-1}+1}^{t_k-1}(1-Y^{\uparrow}_{i,\mu^{-1}(s)}) = 0 \right|{{\cal Y}^{\downarrow}}^{< \tau} = \vec{0} }\nonumber\\
&& \qquad ~~ \mbox{} + \pr{W_{i,\tau} =1}\cdot  \pr{\left. \prod_{s = t_{k-1}+1}^{t_k-1}(1-Y^{\uparrow}_{i,\mu^{-1}(s)}) = 1 \right|{{\cal Y}^{\downarrow}}^{< \tau} = \vec{0} } \label{eq:verif_2ind}\\
&& \quad  =~~ \xi_{i,\tau}+ \frac{p^{\downarrow}_i-\xi_{i,\tau}}{1-\xi_{i,\tau}} \cdot (1-\xi_{i,\tau}) \label{eq:verif_2cond} \\
&& \quad  =~~ p^{\downarrow}_i\nonumber \ .
\end{eqnarray}
Here, equality~\eqref{eq:verif_2ind} proceeds by noting that $W_{i,\tau}$ is independent of $({{\cal Y}^{\downarrow}}^{< \tau}, {{\cal Y}^\uparrow}^{\leq\mu^{-1}(\tau-1)})$. Indeed, by Lemma~\ref{lem:no_future}, we have ${{\cal Y}^{\downarrow}}^{< \tau} = \nu_2({\cal Z}^{< \tau},{\cal W}^{< \tau},{\cal V}^{< \tau},{{\cal Y}^{\uparrow}}^{<  \mu^{-1}(\tau-1)})$. By construction,  $W_{i,\tau}$ is independent of $({\cal Z},{\cal W}^{< \tau},{\cal V},{{\cal Y}^{\uparrow}})$, and thus, in particular, $W_{i,\tau}$ is independent of $(\nu_2({\cal Z}^{< \tau},{\cal W}^{< \tau},{\cal V}^{< \tau},{{\cal Y}^{\uparrow}}^{<  \mu^{-1}(\tau-1)}), {{\cal Y}^\uparrow}^{\leq\mu^{-1}(\tau-1)}) = ({{\cal Y}^{\downarrow}}^{< \tau}, {{\cal Y}^\uparrow}^{\leq\mu^{-1}(\tau-1)})$. To better understand equality~\eqref{eq:verif_2cond}, we observe that
\begin{eqnarray}
&& \pr{\left. \prod_{s = t_{k-1}+1}^{t_k-1}(1-Y^{\uparrow}_{i,\mu^{-1}(s)}) = 1 \right|{{\cal Y}^{\downarrow}}^{< \tau} = \vec{0} } \nonumber \\
&&\quad =~~  \pr{\left. \prod_{s = t_{k-1}+1}^{t_k-1}(1-Y^{\uparrow}_{i,\mu^{-1}(s)}) = 1 \right|{{\cal Y}_i^{\downarrow}}^{< \tau} = \vec{0}} \label{eqn:verif_2cond_1} \\
&&\quad =~~ \frac{\prpar{\{\prod_{s = t_{k-1}+1}^{t_k-1}(1-Y^{\uparrow}_{i,\mu^{-1}(s)}) = 1 \} \wedge \{{{\cal Y}_i^{\downarrow}}^{< \tau} = \vec{0} \} }}{\prpar{{{\cal Y}_i^{\downarrow}}^{< \tau} = \vec{0}}} \nonumber \\
&&\quad =~~ \frac{\prpar{ (\bigwedge_{s = t_{k-1}+1}^{t_k-1} \{ Y^{\uparrow}_{i,\mu^{-1}(s)} = 0 \} ) \wedge \{ {{\cal Y}_i^{\downarrow}}^{\leq t_{k-1}} = \vec{0} \}}}{\prpar{(\bigwedge_{s = t_{k-1}+1}^{ t_k-1} \{ Y^{\downarrow}_{i,s} = 0 \}) \wedge \{ {{\cal Y}_i^{\downarrow}}^{\leq t_{k-1}} = \vec{0} \}}} \label{eqn:verif_2cond_2} \\
&&\quad =~~ \frac{\prpar{\bigwedge_{s = t_{k-1}+1}^{t_k-1} \{ Y^{\uparrow}_{i,\mu^{-1}(s)} = 0 \}}}{\prpar{ \bigwedge_{s = t_{k-1}+1}^{ t_k-1} \{ Y^{\downarrow}_{i,s} = 0 \} }} \label{eqn:verif_2cond_3} \\
&&\quad =~~ \frac{\prpar{\bigwedge_{s = t_{k-1}+1}^{t_k-1} \{ Y^{\uparrow}_{i,\mu^{-1}(s)} = 0 \}}}{\prpar{ \bigwedge_{s = t_{k-1}+1}^{ t_k-1} \{ Y^{\uparrow}_{i,\mu^{-1}(s)} \cdot Z_{i,s} = 0 \} }} \label{eqn:verif_2cond_4} \\
&&\quad =~~ \frac{(1-p_i^{\uparrow})^{t_k - t_{k-1}-1}}{(1-p_i^{\downarrow})^{t_k - t_{k-1}-1}} \nonumber \\
&&\quad =~~ 1 -\xi_{i,\tau} \ . \nonumber
\end{eqnarray}
To justify the non-trivial transitions above, we first note that  equality~\eqref{eqn:verif_2cond_1} holds since the histories of different customers are mutually independent. To better understand equality~\eqref{eqn:verif_2cond_2}, we note that the event $\{Y^{\downarrow}_{i,t_{k-1}+1} = \cdots = Y^{\downarrow}_{i,t_{k}-1} =0\}$ is clearly a subset of $\{Y^{\uparrow}_{i,\mu^{-1}(t_{k-1}+1)} = \cdots = Y^{\uparrow}_{i,\mu^{-1}(t_k-1)} =0\}$, since by case~1 of our coupling construction, we have $Y^{\downarrow}_{i,s} = Y^{\uparrow}_{i,\mu^{-1}(s)} \cdot Z_{i,s}$ for every $t_{k-1}+1 \leq s \leq t_k - 1$. Finally, equality~\eqref{eqn:verif_2cond_3} follows by observing that  $(Y^{\uparrow}_{i,\mu^{-1}(t_{k-1}+1)},\ldots, Y^{\uparrow}_{i,\mu^{-1}(t_k-1)})$ is independent of 
${{\cal Y}^{\downarrow}}^{\leq t_{k-1}}$ based on Lemma~\ref{lem:no_future}, whereas $(Y^{\downarrow}_{i,t_{k-1}+1}, \ldots, Y^{\downarrow}_{i,t_k-1})$ is independent of ${{\cal Y}^{\downarrow}}^{\leq t_{k-1}}$ according to the induction hypothesis. Finally, equality~\eqref{eqn:verif_2cond_4} is obtained by recalling that $Y^{\downarrow}_{i,s} = Y^{\uparrow}_{i,\mu^{-1}(s)} \cdot Z_{i,s}$ for every $t_{k-1}+1 \leq s \leq t_k - 1$.
\end{proof}

\subsection{Designing the policy \texorpdfstring{\boldmath{${\cal S}^{\gamma\uparrow}$}}{}} \label{subsec:up_policy}

We are now ready to introduce the long-anticipated policy ${\cal S}^{\gamma\uparrow}$, attempting to mimic the actions of ${\cal S}^{ *\downarrow }$, while operating in the rounded-up process. We mention in passing that the superscript $\gamma$ reinforces the dependency of our policy on the shifting parameter $\gamma$; this relation will shortly arise by observing that ${\cal S}^{\gamma\uparrow}$ is very much affected by the set of milestones ${\cal M}_{ \gamma }$. For this purpose, we exploit the probabilistic coupling $({{\cal Y}^{\downarrow}}, {\cal Y}^{\uparrow})$, whose specifics have been provided and verified in Sections~\ref{subsec:coupling} and~\ref{subsec:coupling_prop}. At a high level, the policy ${\cal S}^{\gamma\uparrow}$ simulates the evolution of the rounded-down process through the joint distribution of $({{\cal Y}^{\downarrow}}, {\cal Y}^{\uparrow})$. In each stage $t$, we will attempt to pick precisely the same customer as  ${\cal S}^{ *\downarrow }$ does at stage $\mu(t)$, unless s/he has already departed, while skipping milestones along the way.

\paragraph{Information.} We formally specify our policy ${\cal S}^{\gamma\uparrow}$ by induction over the stage index $t\geq 1$. To this end, let $(t, {\cal A}^{{\cal Y}^{\uparrow}}_t)$ be the state reached by ${\cal S}^{\gamma\uparrow}$ at the beginning of stage $t$ of the rounded-up process, where ${\cal A}^{{\cal Y}^{\uparrow}}_t \subseteq [n]$ stands for the subset of remaining customers at that time.  As explained below, the policy ${\cal S}^{\gamma\uparrow}$ either picks the next customer to be served (out of ${\cal A}^{{\cal Y}^{\uparrow}}_t$), or chooses not to serve any customer (indicated by $\perp$). To arrive at this decision, our policy clearly has access to its own history of departures in the rounded-up process, ${{\cal Y}^{\uparrow}}^{\leq t-1}$. In addition, we further assume that ${\cal S}^{\gamma\uparrow}$ has already sampled the ``corresponding'' history of the rounded-down process, ${{\cal Y}^{\downarrow}}^{\leq \mu(t-1)-1}$. By reading through the upcoming paragraph, it is easy to verify that this assumption is preserved throughout our inductive construction.

\paragraph{Sampling.} Having already sampled ${{\cal Y}^{\downarrow}}^{\leq \mu(t-1)-1}$, as a preliminary operation, we explain how the policy ${\cal S}^{\gamma\uparrow}$ draws the additional departure outcomes $(Y_{i,s}^{\downarrow})_{i\in [n]}$ in stages $s=\mu(t-1),\ldots, \mu(t)-1$ of the rounded-down process. To this end, we distinguish between two cases, depending on whether stage $\mu(t)-1$ is regular or not.
\begin{enumerate}
\item {\em When $\mu(t)-1 \notin M_{\gamma}$:}	In this case, $\mu(t)-1 = \mu(t-1)$, meaning that we should be sampling the Bernoulli outcomes $\{Y^{\downarrow}_{i,\mu(t)-1}\}_{i\in [n]}$ of only one stage, $\mu(t)-1$. Since the latter stage is regular, we simply employ equation~\eqref{eq:regular} to compute   $Y^{\downarrow}_{i,\mu(t)-1} = Y^{\downarrow}_{i,\mu(t-1)} = Y^{\uparrow}_{i,t-1} \cdot Z_{i,\mu(t)-1}$ for each customer $i\in [n]$. Here, we should sample the auxiliary random variable $Z_{i,\mu(t)-1}$, noting that $Y^{\uparrow}_{i,t-1}$ is already known as part of the available history ${{\cal Y}^{\uparrow}}^{\leq t-1}$.  

\item {\em When $\mu(t)-1 \in M_{\gamma}$:}	In this case, we have $\mu(t)-1 = \mu(t-1) + 1$, meaning that we should actually be sampling the Bernoulli outcomes $\{Y^{\downarrow}_{i,\mu(t-1)}\}_{i\in [n]}$ and $\{Y^{\downarrow}_{i,\mu(t)-1}\}_{i\in [n]}$, corresponding to two successive stages, $\mu(t-1)$ and $\mu(t)-1$. Noting that stage $\mu(t-1)$ is regular, its sampling procedure is identical to the one described in case~1 above. By contrast, stage $\mu(t)-1$ is a milestone, and as such, there exists a unique index $k\geq 1$ for which $\mu(t)-1 = t_k$. Here, we employ equation~\eqref{eq:x-iu}:
\begin{eqnarray*} 
Y^{\downarrow}_{i,t_k} ~~=~~
    \begin{dcases}
    1 - \left(1-W_{i,t_k}\right)\cdot \prod_{s= {t}_{k-1}+1}^{{t}_{k}-1}(1-Y^{\uparrow}_{i,\mu^{-1}(s)}) \qquad & \text{if } Y^{\downarrow}_{i,1} = \cdots = Y^{\downarrow}_{i,t_k-1} = 0  \\
    {V}_{i,t_k} & \text{otherwise}
    \end{dcases} 	
    \end{eqnarray*}
To utilize this equation, we should sample the auxiliary variables $W_{i,\tau}$ and $V_{i,\tau}$, noting that the realizations of $Y^{\uparrow}_{i,\mu^{-1}(t_{k-1}+1)},\ldots,Y^{\uparrow}_{i,\mu^{-1}(t_{k}-1)}$  are already known as part of the available history ${{\cal Y}^{\uparrow}}^{\leq t-1}$, since  $\mu^{-1}(t_{k} - 1) = \mu^{-1}(\mu(t-1)) = t-1$. Additionally, the realizations of $Y^{\downarrow}_{i,1}, \ldots, Y^{\downarrow}_{i,t_k-2}$ are known from the previous stage, whereas  $Y^{\downarrow}_{i,t_k-1} = Y^{\downarrow}_{i,\mu(t-1)}$ has just been computed in the current stage. 
\end{enumerate}

\paragraph{Service action.} At this point, the history ${{\cal Y}^{\uparrow}}^{\leq t-1}$ of departures in the rounded-up process has been observed, and concurrently, the history ${{\cal Y}^{\downarrow}}^{\leq \mu(t)-1}$ of departures in the rounded-down process has already been sampled. Given this information, our policy can be specified with the exact knowledge of how the policy ${\cal S}^{ *\downarrow }$ operates until stage $\mu(t)$. In particular,  we can determine the subset ${\cal A}^{{{\cal Y}^{\downarrow}}}_{\mu(t)}$ of available customers in the rounded-down process at the beginning of stage $\mu(t)$. On the other hand, when employing our policy ${\cal S}^{\gamma\uparrow}$ in the rounded-up process, stage $t$ begins with ${\cal A}^{{\cal Y}^{\uparrow}}_t$ as the set of currently available customers, which could very well be different from ${\cal A}^{{{\cal Y}^{\downarrow}}}_{\mu(t)}$. That said, our policy ${\cal S}^{\gamma\uparrow}$ examines whether customer ${\cal S}^{ *\downarrow }(\mu(t), {\cal A}^{{{\cal Y}^{\downarrow}}}_{\mu(t)})$, who is precisely the one served by ${\cal S}^{ *\downarrow }$ in stage $\mu(t)$ of the rounded-down process, is still available or not. With an affirmative answer, this customer is served; otherwise, ${\cal S}^{\gamma\uparrow}$ does not serve any customer and proceeds to the next stage. To summarize,  
\begin{equation} \label{eqn:service_dec}
{\cal S}^{\gamma\uparrow}(t, {\cal A}^{{\cal Y}^{\uparrow}}_t) ~~=~~
\begin{dcases}
{\cal S}^{ *\downarrow }(\mu(t), {\cal A}^{{{\cal Y}^{\downarrow}}}_{\mu(t)}) \qquad  & \text{if } {\cal S}^{ *\downarrow }(\mu(t), {\cal A}^{{{\cal Y}^{\downarrow}}}_{\mu(t)})\in {\cal A}^{{\cal Y}^{\uparrow}}_t \\
\perp & \text{otherwise}
\end{dcases}
\end{equation}

It is important to point out that, conditional on being in state $(t, {\cal A}^{{\cal Y}^{\uparrow}}_t) $, the chosen action  ${\cal S}^{\gamma\uparrow}(t, {\cal A}^{{\cal Y}^{\uparrow}}_t)$ is random,  as it depends on the random set ${\cal A}^{{{\cal Y}^{\downarrow}}}_{\mu(t)}$ of available customers in the simulated rounded-down process. Our sampling procedure ensures that the resulting policy ${\cal S}^{\gamma\uparrow}$ is non-anticipative, since it only makes use of the information available so far with respect to the rounded-up process. This claim is formalized by Lemma~\ref{lem:no_future}, which implies that the information $({{\cal Y}^{\uparrow}}^{\leq t-1},{{\cal Y}^{\downarrow}}^{\leq \mu(t)-1})$, on which our service decision ${\cal S}^{\gamma\uparrow}(t, {\cal A}^{{\cal Y}^{\uparrow}}_t)$ was based, is independent of all future information ${{\cal Y}^{\uparrow}}^{> t-1}$ regarding the rounded-up process. 

\subsection{Concluding Theorem~\ref{thm:rev_diff_direction}}  \label{subsec:analysis}

Following the outline of Section~\ref{subsec:proof_outline}, we move on to establish its required performance guarantee, showing that with an appropriate choice of the shifting parameter $\gamma$, the expected reward $\rev^{ (p^{ \uparrow }) }( {\cal S}^{\gamma \uparrow} )$ of our policy ${\cal S}^{\gamma \uparrow}$ in the rounded-up process near matches the corresponding reward ${\cal R}^{(p^{\downarrow})}({\cal S}^{ * \downarrow})$ of ${\cal S}^{ * \downarrow}$ in the rounded-down policy.

\paragraph{The probability of duplicating $\bs{{\cal S}^{ * \downarrow }}$.} For this purpose, our key observation is that the policy ${\cal S}^{\gamma\uparrow}$ is capable of implementing the service decisions made by ${\cal S}^{ * \downarrow}$ with high probability, as we proceed to show next. To this end, we first argue that the customer served by ${\cal S}^{ * \downarrow}$ in stage $\mu(t)$ is also available in the rounded-up process at time $\mu^{-1}(t_{k-1}+1)$ with probability 1, where $t_{k-1}$ is the milestone that immediately precedes $\mu(t)$. 

\begin{lemma} \label{lem:lb-milestone}
For every stage $t\geq 1$ and for every customer $i\in [n]$,
\[ \prsub{({\cal Y}^{\downarrow},{\cal Y}^\uparrow)}{\left. i\in {\cal A}^{{\cal Y}^\uparrow}_{\mu^{-1}(t_{k-1}+1)} \right| {\cal S}^{ * \downarrow} (\mu(t), {\cal A}^{{\cal Y}^{\downarrow}}_{\mu(t)} ) = i } ~~=~~ 1 \ , \]
where $k$ be the unique milestone index for which $\mu(t) \in [t_{k-1} + 1,t_k-1]$.
\end{lemma}
\begin{proof}
We begin by noting that, conditional on the event $\{{\cal S}^{ * \downarrow}(\mu(t), {\cal A}^{{\cal Y}^{\downarrow}}_{\mu(t)}) = i\}$, we clearly have $Y^{\downarrow}_{i,1} = \cdots = Y^{\downarrow}_{i,\mu(t)-1} = 0$ with probability 1. Otherwise, customer $i$ would not have been available at stage $\mu(t)$ of the rounded-down process, meaning that s/he could not have been served by ${\cal S}^{ * \downarrow}$ at state $(\mu(t), {\cal A}^{{\cal Y}^{\downarrow}}_{\mu(t)} )$. Given this observation, the desired result follows from the next two claims, both conditional on $\{{\cal S}^{ * \downarrow}(\mu(t), {\cal A}^{{\cal Y}^{\downarrow}}_{\mu(t)}) = i\}$:
\begin{itemize}
    \item {\em Customer $i$ has not departed yet, i.e., $Y^{\uparrow}_{i,1} = \cdots = Y^{\uparrow}_{i,\mu^{-1}(t_{k-1}+1)-1} = 0$}. Suppose by contradiction that there exists some stage $s \leq \mu^{-1}(t_{k-1}+1)-1$ for which $Y^{\uparrow}_{i,s} =1$, and let $\hat{k}$ be the milestone index for which $\mu(s) \in [t_{\hat{k}-1}+1,t_{\hat{k}}-1]$. Since $t_{\hat{k}}$ is a milestone, by the upper sub-case of equation~\eqref{eq:x-iu} in our coupling construction, we infer that $Y^{\downarrow}_{i,t_{\hat{k}}} = 1$. However, since $t_{\hat{k}} \leq t_{k-1} < \mu(t)$, we have just obtained a contradiction to the observation that $Y^{\downarrow}_{i,1} = \cdots = Y^{\downarrow}_{i,\mu(t)-1} = 0$.
    
    \item {\em Customer $i$ has not been served yet, i.e., ${\cal S}^{\gamma \uparrow}(s, {\cal A}^{{\cal Y}^\uparrow}_{s})\neq i$ for every stage $s \leq t-1$}. It is worth pointing out that this claim is stronger than what is required to show that $i\in {\cal A}^{{\cal Y}^\uparrow}_{\mu^{-1}(t_{k-1}+1)}$, since $\mu^{-1}(t_{k-1}+1) \leq t$; however, it will become handy in our subsequent analysis. Again, suppose by contradiction that there exists some stage $s \leq t-1$ for which ${\cal S}^{\gamma \uparrow}(s, {\cal A}^{{\cal Y}^\uparrow}_{s}) = i$. As explained in Section~\ref{subsec:up_policy} (see equation~\eqref{eqn:service_dec}), a necessary condition for the policy ${\cal S}^{\gamma \uparrow}$ to make this decision is that ${\cal S}^{ *\downarrow }(\mu(s), {\cal A}^{{{\cal Y}^{\downarrow}}}_{\mu(s)}) = i$. However, the latter decision of ${\cal S}^{ *\downarrow }$ implies in particular that customer $i$ is no longer available starting at stage $\mu(s)+1 \leq \mu(t)$ of the rounded-down process, in contradiction to having ${\cal S}^{ * \downarrow}(\mu(t), {\cal A}^{{\cal Y}^{\downarrow}}_{\mu(t)}) = i$.
\end{itemize}
\end{proof}

Building on Lemma~\ref{lem:lb-milestone}, we prove that customer ${\cal S}^{ *\downarrow } (\mu(t), {\cal A}^{{{\cal Y}^{\downarrow}}}_{\mu(t)} )$ is also available at time $t$ with conditional probability at least $1-\eps$.

\begin{lemma} \label{lem:stay}
For every stage $t\geq 1$ and for every customer $i\in [n]$,
\begin{eqnarray*}
\prsub{({{\cal Y}^{\downarrow}},{\cal Y}^{\uparrow})}{ i\in {\cal A}^{{\cal Y}^{\uparrow}}_t \left| {\cal S}^{ *\downarrow } (\mu(t), {\cal A}^{{{\cal Y}^{\downarrow}}}_{\mu(t)} ) = i \right. } ~~\geq~~ 1-\eps \ .
\end{eqnarray*}
\end{lemma}
\begin{proof}
For notational convenience, since all probabilities in this proof are computed with respect to the coupling $({\cal Y}^{\downarrow},{\cal Y}^{\uparrow})$, we drop this reference below. Letting $k$ be the unique milestone index for which $\mu(t) \in [t_{k-1} + 1,t_k-1]$, we begin by observing that 
\begin{eqnarray}
&&\pr{i\in {\cal A}^{{\cal Y}^{\uparrow}}_t \left| {\cal S}^{ * \downarrow} (\mu(t), {\cal A}^{{\cal Y}^{\downarrow}}_{\mu(t)} ) = i \right. } \nonumber \\
&& \quad =~~ \pr{ \{ i\in {\cal A}^{{\cal Y}^{\uparrow}}_t \} \wedge \{ i\in {\cal A}^{{\cal Y}^\uparrow}_{\mu^{-1}(t_{k-1}+1)} \} \left| {\cal S}^{ * \downarrow} (\mu(t), {\cal A}^{{\cal Y}^{\downarrow}}_{\mu(t)} ) = i \right. } \label{eq:proof_lem_stay_eq0} \\
&& \quad =~~  \pr{\left. i\in {\cal A}^{{\cal Y}^\uparrow}_{\mu^{-1}(t_{k-1}+1)} \right| {\cal S}^{ * \downarrow} (\mu(t), {\cal A}^{{\cal Y}^{\downarrow}}_{\mu(t)} ) = i } \nonumber \\
&& \qquad~~ \mbox{} \cdot \pr{ i\in {\cal A}^{{\cal Y}^{\uparrow}}_t \left| \{ {\cal S}^{ * \downarrow} (\mu(t), {\cal A}^{{\cal Y}^{\downarrow}}_{\mu(t)} ) = i \} \wedge\{ i \in {\cal A}^{{\cal Y}^{\uparrow}}_{\mu^{-1}(t_{k-1}+1)} \} \right. } \nonumber\\
&& \quad =~~ \pr{ Y^{\uparrow}_{i,\mu^{-1}(t_{k-1}+1)} = \cdots =  Y^{\uparrow}_{i,t-1} = 0 \left| \{ {\cal S}^{ * \downarrow} (\mu(t), {\cal A}^{{\cal Y}^{\downarrow}}_{\mu(t)} ) = i \} \wedge\{ i \in {\cal A}^{{\cal Y}^{\uparrow}}_{\mu^{-1}(t_{k-1}+1)} \} \right. } \ . \label{eq:proof_lem_stay_eq05} 
\end{eqnarray}
Here, equality~\eqref{eq:proof_lem_stay_eq0} holds since, when customer $i$ is available at stage $t$ of the rounded-up process (i.e., $i\in {\cal A}^{{\cal Y}^{\uparrow}}_t$), s/he is necessarily available at any earlier stage, implying in particular that $i\in {\cal A}^{{\cal Y}^\uparrow}_{\mu^{-1}(t_{k-1}+1)}$. Equality~\eqref{eq:proof_lem_stay_eq05} follows from Lemma~\ref{lem:lb-milestone}, in conjunction with the observation that, conditional on $\{ {\cal S}^{ * \downarrow} (\mu(t), {\cal A}^{{\cal Y}^{\downarrow}}_{\mu(t)} ) = i \}\wedge\{ i \in {\cal A}^{{\cal Y}^{\uparrow}}_{\mu^{-1}(t_{k-1}+1)} \}$, the policy ${\cal S}^{\gamma\uparrow}$ does not serve customer $i$ before stage $t$, and thus, this customer is available at stage $t$ of the rounded-up process if and only if s/he does not depart in any of the stages $\mu^{-1}(t_{k-1}+1), \ldots, t-1$, namely, $Y^{\uparrow}_{i,\mu^{-1}(t_{k-1}+1)} = \cdots =  Y^{\uparrow}_{i,t-1} = 0$. 

To lighten our notation, let ${\cal H}^\uparrow_i$ be the collection of all sample path realizations of the process ${{\cal Y}_i^\uparrow}^{\leq \mu^{-1}(t_{k-1}+1)-1}$ that occur with non-zero probability, conditional on the event $\{ {\cal S}^{ * \downarrow} (\mu(t), {\cal A}^{{{\cal Y}^{\downarrow}}}_{\mu(t)}) = i \} \wedge \{ i \in {\cal A}^{{\cal Y}^{\uparrow}}_{\mu^{-1}(t_{k-1}+1)} \}$. We further note that $\prpar{{{\cal Y}_i^{\downarrow}}^{\leq \mu(t-1)} = h^\downarrow_i |  {\cal S}^{ * \downarrow} (\mu(t), {\cal A}^{{{\cal Y}^{\downarrow}}}_{\mu(t)} ) = i } = 0$ for all realizations $h^\downarrow_i \neq \vec{0}$ of the process ${{\cal Y}_i^{\downarrow}}^{\leq \mu(t-1)}$, or otherwise, customer $i$ would have departed prior to stage $\mu(t)$ of the rounded-down process. We proceed by arguing that the right-hand-side of equality~\eqref{eq:proof_lem_stay_eq05} is indeed lower-bounded by $1 - \eps$. To this end, note that
\begin{eqnarray} 
&& \pr{ Y^{\uparrow}_{i,\mu^{-1}(t_{k-1}+1)} = \cdots =  Y^{\uparrow}_{i,t-1} = 0 \left| \{ {\cal S}^{ * \downarrow} (\mu(t), {\cal A}^{{\cal Y}^{\downarrow}}_{\mu(t)} ) = i \} \wedge\{ i \in {\cal A}^{{\cal Y}^{\uparrow}}_{\mu^{-1}(t_{k-1}+1)} \} \right. } \nonumber \\
&& \quad =~~ \sum_{h^\uparrow_i \in {\cal H}^\uparrow_i } \pr{{{\cal Y}^{\uparrow}_i}^{\leq \mu^{-1}(t_{k-1}+1)-1} = h^\uparrow_i \left| \{ {\cal S}^{ * \downarrow} (\mu(t), {\cal A}^{{{\cal Y}^{\downarrow}}}_{\mu(t)} ) = i \} \wedge \{ i \in {\cal A}^{{\cal Y}^{\uparrow}}_{\mu^{-1}(t_{k-1}+1)} \}\right.} \nonumber \\
&& \qquad \qquad \qquad \mbox{} \cdot \pr{ Y^{\uparrow}_{i,\mu^{-1}(t_{k-1}+1)} = \cdots =  Y^{\uparrow}_{i,t-1} = 0 \left| ({{\cal Y}^{\uparrow}_i}^{\leq \mu^{-1}(t_{k-1}+1)-1}, {{\cal Y}_i^{\downarrow}}^{\leq \mu(t-1)}) = (h^\uparrow_i, \vec{0}) \right. }  \nonumber  \\
&& \quad =~~ \pr{ Y^{\uparrow}_{i,\mu^{-1}(t_{k-1}+1)} = \cdots = Y^{\uparrow}_{i,t-1} =0 \left| Y^{\downarrow}_{i,t_{k-1}+1} = \cdots = Y^{\downarrow}_{i,\mu(t-1)} =0 \right. } \label{eq:proof_lem_stay_eq4} \\
&& \quad =~~ \frac{ \prpar{ Y^{\uparrow}_{i,\mu^{-1}(t_{k-1}+1)} = \cdots = Y^{\uparrow}_{i,t-1} =0 } }{ \prpar{ Y^{\downarrow}_{i,t_{k-1}+1} = \cdots = Y^{\downarrow}_{i,\mu(t-1)} =0 } } \label{eq:proof_lem_stay_eq45} \\
&& \quad =~~ \frac{(1-p_i^{\uparrow})^{t - t_{k-1}-1}}{(1-p_i^{\downarrow})^{t - t_{k-1}-1}} \label{eq:proof_lem_stay_eq5}	\\
&& \quad \geq~~ 1-\eps	\ . \label{eq:proof_lem_stay_eq6}
\end{eqnarray}
Here, equality~\eqref{eq:proof_lem_stay_eq4} proceeds from Lemma~\ref{lem:no_future}, which implies in particular that $(Y^{\uparrow}_{i,\mu^{-1}(t_{k-1}+1)}, \ldots, Y^{\uparrow}_{i,t-1})$ is independent of $({{\cal Y}^{\uparrow}}^{\leq \mu^{-1}(t_{k-1}+1)-1},{{\cal Y}^{\downarrow}}^{\leq t_{k-1}})$. Equality~\eqref{eq:proof_lem_stay_eq45} holds since $Y^{\downarrow}_{i,\tau} = Y^{\uparrow}_{i,\mu^{-1}(\tau)} \cdot Z_{i,\tau}$ for every regular stage $\tau$, as prescribed in equation~\eqref{eq:regular}. To obtain equality~\eqref{eq:proof_lem_stay_eq5}, we recall that ${\cal Y}^{\uparrow} \sim {{\cal X}^{\uparrow}}$ and ${{\cal Y}^{\downarrow}}\sim {{\cal X}^{\downarrow}}$, as shown in Section~\ref{subsec:coupling_prop}. Therefore, $Y^{\uparrow}_{i,\mu^{-1}(t_{k-1}+1)}, \ldots, Y^{\uparrow}_{i,t-1}$ and 
$Y^{\downarrow}_{i,t_{k-1}+1}, \ldots, Y^{\downarrow}_{i,\mu(t-1)}$ are independent Bernoulli random variables with success probabilities $p^\uparrow_i$ and $p^\downarrow_i$, respectively. Finally, inequality~\eqref{eq:proof_lem_stay_eq6} follows by instantiating Lemma~\ref{lem:comp_prob_short} with $\Delta = t - t_{k-1}-1 \leq \frac{1}{\eps}-1$.
\end{proof}

\paragraph{Relating the expected rewards $\bs{\rev^{ (p^{ \uparrow }) }( {\cal S}^{\gamma \uparrow} )}$ and $\bs{{\cal R}^{(p^{\downarrow})}({\cal S}^{ * \downarrow})}$.} In light of the preceding discussion, it follows that in the rounded-up process, the policy ${\cal S}^{\gamma\uparrow}$ guarantees an expected reward of 
\begin{eqnarray}
{\cal R}^{(p^{\uparrow})} ({\cal S}^{\gamma\uparrow} ) & = & \sum_{t \geq 1} \sum_{i \in [n]} \prsub{{\cal X}^\uparrow}{{\cal S}^{\gamma\uparrow}(t, {\cal A}^{{\cal X}^\uparrow}_t) = i} \cdot r_i \nonumber \\
& = & \sum_{t \geq 1} \sum_{i \in [n]} \prsub{{\cal Y}^\uparrow}{{\cal S}^{\gamma\uparrow}(t, {\cal A}^{{\cal Y}^\uparrow}_t) = i} \cdot r_i \label{eqn:bound_rev_gamma_0}	\\
& = & \sum_{t \geq 1} \sum_{i \in [n]} \prsub{({{\cal Y}^{\downarrow}},{\cal Y}^\uparrow)}{ \{i \in {\cal A}^{{\cal Y}^\uparrow}_t\} \wedge \{{\cal S}^{ * \downarrow}(\mu(t),{\cal A}^{{{\cal Y}^{\downarrow}}}_{\mu(t)}) = i \}} \cdot r_i \label{eqn:bound_rev_gamma_1}\\
& = & \sum_{t \geq 1} \sum_{i \in [n]} \prsub{{{\cal Y}^{\downarrow}}}{ {\cal S}^{ * \downarrow}(\mu(t),{\cal A}^{{{\cal Y}^{\downarrow}}}_{\mu(t)}) = i} \cdot \prsub{({{\cal Y}^{\downarrow}},{\cal Y}^\uparrow)}{ i \in {\cal A}^{\cal Y}_t \left| {\cal S}^{ * \downarrow}(\mu(t),{\cal A}^{{{\cal Y}^{\downarrow}}}_{\mu(t)}) = i \right.}\cdot r_i \nonumber	\\
& \geq & (1 - \eps) \cdot \sum_{t \geq 1} \sum_{i \in [n]} \prsub{{{\cal Y}^{\downarrow}}}{ {\cal S}^{ * \downarrow}(\mu(t),{\cal A}^{{{\cal Y}^{\downarrow}}}_{\mu(t)}) = i}\cdot r_i \label{eqn:bound_rev_gamma_2}	\\
& = & (1 - \eps) \cdot \sum_{t \geq 1} \sum_{i \in [n]} \bbInd \left[ t \notin {\cal M}_{\gamma} \right] \cdot \prsub{{\cal Y}^{\downarrow}}{ {\cal S}^{ * \downarrow}(t,{\cal A}^{{\cal Y}^{\downarrow}}_{t}) = i}\cdot r_i \nonumber \\
& = & (1 - \eps) \cdot \sum_{t \geq 1} \sum_{i \in [n]} \bbInd \left[ t \notin {\cal M}_{\gamma} \right] \cdot \prsub{{{\cal X}^{\downarrow}}}{ {\cal S}^{ * \downarrow}\left(t,{\cal A}^{{{\cal X}^{\downarrow}}}_{t}\right) = i}\cdot r_i \ . \label{eqn:bound_rev_gamma_4}
\end{eqnarray}
Here, equalities~\eqref{eqn:bound_rev_gamma_0} and~\eqref{eqn:bound_rev_gamma_4} are respectively obtained by recalling that ${{\cal Y}^{\uparrow}}\sim {{\cal X}^{\uparrow}}$ and ${{\cal Y}^{\downarrow}}\sim {{\cal X}^{\downarrow}}$, as shown in Section~\ref{subsec:coupling_prop}. Equality~\eqref{eqn:bound_rev_gamma_1} holds since, by the way our policy ${\cal S}^{\gamma\uparrow}$ is defined in Section~\ref{subsec:up_policy}, we have ${\cal S}^{\gamma\uparrow}(t, {\cal A}^{\cal Y}_t) = i$ if and only if customer $i$ is the one served by ${\cal S}^{ * \downarrow}$ at stage $\mu(t)$ (namely, ${\cal S}^{ * \downarrow}(\mu(t),{\cal A}^{{\cal Y}^{\downarrow}}_{\mu(t)}) = i$) and this customer is still available in the rounded-up process (i.e., $i \in {\cal A}^{{\cal Y}^\uparrow}_t$). Inequality~\eqref{eqn:bound_rev_gamma_2} is a direct implication of Lemma~\ref{lem:stay}. 

\paragraph{Picking the shifting parameter.} We are now ready to specify how the shifting parameter $\gamma$ is determined: We simply pick $\gamma$ uniformly at random out of $\{ 1, \ldots, \frac{1}{\eps} \}$. In this case, taking expectations over the randomness in $\gamma$ on both sides of inequality~\eqref{eqn:bound_rev_gamma_4}, we observe that
\begin{eqnarray*}
\exsub{ \gamma }{ {\cal R}^{(p^{\uparrow})}({\cal S}^{\gamma\uparrow}) } & \geq & (1 - \eps) \cdot \sum_{t \geq 1} \sum_{i \in [n]} \prsub{ \gamma }{ t \notin {\cal M}_{\gamma} } \cdot \prsub{{{\cal X}^{\downarrow}}}{ {\cal S}^{ * \downarrow}(t,{\cal A}^{{{\cal X}^{\downarrow}}}_{t}) = i}\cdot r_i \\
& \geq & (1 - 2\eps) \cdot \sum_{t \geq 1} \sum_{i \in [n]} \prsub{{{\cal X}^{\downarrow}}}{ {\cal S}^{ * \downarrow}(t,{\cal A}^{{{\cal X}^{\downarrow}}}_{t}) = i}\cdot r_i \\
& = & (1-2\eps) \cdot {\cal R}^{(p^{\downarrow})}({\cal S}^{ * \downarrow})\ ,
\end{eqnarray*}
where the second inequality holds since $\prparsub{ \gamma }{ t \notin {\cal M}_{\gamma} } = 1 - \eps$ for every $t \geq 1$. Consequently, it follows that there exists a value of $\gamma$ for which the deterministic policy ${\cal S}^{\gamma\uparrow}$ attains an expected reward of at least $(1-2\eps) \cdot {\cal R}^{(p^{\downarrow})}({\cal S}^{ * \downarrow})$, which is precisely what we were opting to prove.
\section{Concluding Remarks}

We conclude this paper with a number of open questions for future research, all appearing to be highly non-trivial. As explained below, these prospective directions are intended to further narrow the residual gap between the knowns and unknowns, as well as to highlight several extensions where our current understanding is still incomplete.

\paragraph{Improved approximations.} Building on our main result, a quasi-polynomial time approximation scheme for adaptively serving impatient customers, one intriguing direction would be to examine whether this approach can be enhanced and sharpened, potentially ending up with a true PTAS. An outcome of this nature would be a very unique result, since polynomial-time approximation schemes are generally unknown for online stochastic matching problems, even in stylized settings. In this regard, the probabilistic analysis of Theorem~\ref{thm:rev_diff_direction} may play an important role, as it allows us to reduce the number of distinct customer classes to $O_\eps(\log^2 n)$. However, in order to arrive at a polynomial-sized dynamic program, one still needs to propose a fundamentally different encoding of its state space. Thus, we believe that improving on our main result would require substantial developments, perhaps including algorithmic ideas such as batched service decisions or more coarse state pruning methods. On a different front, it would be interesting to investigate whether the polynomial-time $0.709$-approximation due to \cite{CyganEGMS13} can be meaningfully breached. To this end, one potential avenue is that of exploring whether our dynamic programming ideas can be combined with their LP-based methods to obtain an improved approximation guarantee.

\paragraph{Model extensions.} Given the rare occurrence of a near-polynomial-time approximation scheme for a stochastic matching-type problem, one could explore whether our findings can be applied to closely related models. For example, one key feature of the model formulation studied in this paper is that its initial stock of customers is never replenished. Hence, it would be interesting to consider extended formulations, where new customers could be arriving over time, similarly to some of the stochastic models reviewed in Section~\ref{subsec:related}. As noted by \citet[Sec.~6]{CyganEGMS13}, competitive ratios better than $1-\frac{1}{e}$ can be obtained in certain parametric regimes:
\begin{quote}
``{\em Can we give algorithms for the case of user arrivals? (For this last problem, it is not difficult to adapt the algorithm from \citep{Jez11} to give an improvement over $1 - 1/e$ when the survival probabilities are bounded away from 1)}.''
\end{quote}
Yet another interesting extension is obtained by allowing multiple customers to be served in each stage. As a concrete example, suppose that each stage $t$ is associated with a capacity of $k_t$, standing for the maximum number of customers that can be served at that time. In such settings, it would be interesting to examine whether one could improve on the straightforward $(1 - \frac{ 1 }{ e })$-approximation that follows from the work of \cite{AggarwalGKM11} on vertex-weighted online matching, by a reduction along the lines of Section~\ref{subsec:known-results}.

\paragraph{Acknowledgements.}
I am deeply indebted to Ali Aouad (London Business School) for his continuous help throughout this project, especially in regard to the coupling-based proof of Theorem~\ref{thm:rev_diff_direction}. I would like to thank Sahil Singla (Georgia Tech) for introducing me to the work of \cite{CyganEGMS13}. Finally, I am grateful to Will Ma (Columbia University) for fruitful discussions during early stages of this work.

\bibliographystyle{plainnat}
\bibliography{BIB-Impatient}

\appendix
\section{Additional Proofs}

\subsection{Proof of Lemma~\ref{lem:lazy_exists}} \label{app:proof_lem_lazy_exists}

For ease of analysis, rather that forcing adaptive service policies to serve one of the available customers in ${\cal A}_t$ at any stage $t$, we allow an additional action, where instead some customer $i \in {\cal A}_t$ is ``marked''. To capture such actions, we keep utilizing our current notation ${\cal S}(t,{\cal A}_t) = i$ to indicate that customer $i$ is served, whereas ${\cal S}(t,{\cal A}_t) = \text{``mark $i$''}$ corresponds to marking this customer. It is important to point out that, when a customer is marked at stage $t$, s/he is by no means eliminated from the system. Rather, as explained below, our policy simply differentiates between marked and unmarked customers in ${\cal A}_t$, which will be denoted by $\mymark_t$ and $\myunmark_t$, respectively. For clarity, we explicitly specify these two sets within our state description, $(t, \myunmark_t, \mymark_t )$. As mentioned in Section~\ref{subsec:reduce_average}, extended policies that possibly decide not to serve any customer at any given stage can easily be translated back to our standard notion of an adaptive service policy with no deterioration in their expected revenue.

\paragraph{Constructing $\bs{{\cal S}^{\myco}}$.} In what follows, we define an extended policy ${\cal S}^{\myco}$ that imitates the optimal policy ${\cal S}^*$, except for dealing with quitters and stickers. Specifically, noting that ${\cal S}^*$ terminates within at most $n$ stages, our policy will be inductively constructed as follows.
\begin{itemize}
    \item {\em Stage 1:} In this case, ${\cal S}^{\myco}$ mimics ${\cal S}^*$, unless the latter serves a sticker, in which case we avoid serving this customer and mark him/her instead. In other words,
    \[ {\cal S}^{\myco}(1,[n],\emptyset) ~~=~~
    \begin{dcases}
    {\cal S}^{*}(1,[n]) \qquad & \text{if ${\cal S}^{*}(1,[n]) \in {\cal C}_{\myavg} \cup {\cal C}_{\myquit}$} \\
    \text{mark } {\cal S}^{*}(1,[n]) \qquad & \text{if ${\cal S}^{*}(1,[n]) \in {\cal C}_{\mystick}$}
    \end{dcases} \]
    
    \item {\em Stages $t = 2, \ldots, n$:} Here, ${\cal S}^{\myco}$ inspects the service decision ${\cal S}^*$ makes for currently unmarked customers. When the intended customer is average, s/he is served; otherwise, we avoid serving this customer and mark him/her instead. That is, 
    \[ {\cal S}^{\myco}(t, \myunmark_t, \mymark_t ) ~~=~~
    \begin{dcases}
    {\cal S}^{*}(t,\myunmark_t) \qquad & \text{if ${\cal S}^{*}(t,\myunmark_t) \in {\cal C}_{\myavg}$} \\
    \text{mark } {\cal S}^{*}(t,\myunmark_t) \qquad & \text{if ${\cal S}^{*}(t,\myunmark_t) \in {\cal C}_{\mystick} \cup {\cal C}_{\myquit}$}
    \end{dcases} \]
    
    \item {\em Stages $t = n+1, \ldots, 2n$:} In these stages, our policy picks the highest-reward sticker who is still available and unmarked. For this purpose, assuming without loss of generality that stickers are indexed in weakly-decreasing order of their rewards, we have
    \[ {\cal S}^{\myco}(t, \myunmark_t, \mymark_t ) ~~=~~
    \min ( \myunmark_t \cup {\cal C}_{\mystick}) \ . \]
\end{itemize}
By going through the construction above, one can easily verify that the policy ${\cal S}^{\myco}$ is indeed class-ordered. Yet another important observation is that, letting $(\myunmark_t, \mymark_t )$ and ${\cal A}_t^*$ be the random sets of available customers at stage $t$ when respectively employing the service policies ${\cal S}^{\myco}$ and ${\cal S}^*$, the sets $\myunmark_t$ and ${\cal A}_t^*$ are in fact identically distributed, for every $1 \leq t \leq n$. Moving forward, this property will become very useful.

\paragraph{Bounding the expected reward of $\bs{{\cal S}^{\myco}}$.} In order to relate the expected rewards $\rev({\cal S}^{\myco})$ and $\rev({\cal S}^{*})$, we first decompose these quantities into the individual contributions of the underlying customers, grouped by their class. More specifically, we have on the one hand
\begin{eqnarray*}
{\cal R}({\cal S}^*) & = & \underbrace{ \sum_{t \in [n]} \sum_{i \in {\cal C}_{\myavg}} \pr{{\cal S}^*(t, {\cal A}_t^*) = i} \cdot r_i }_{ {\cal R}_{ \myavg }({\cal S}^*) } + \underbrace{ \sum_{t \in [n]} \sum_{i \in {\cal C}_{\myquit}} \pr{{\cal S}^*(t, {\cal A}_t^*) = i} \cdot r_i }_{ {\cal R}_{ \myquit }({\cal S}^*) }\\
&& \mbox{} + \underbrace{ \sum_{t \in [n]} \sum_{i \in {\cal C}_{\mystick}} \pr{{\cal S}^*(t, {\cal A}_t^*) = i} \cdot r_i }_{ {\cal R}_{ \mystick }({\cal S}^*) } \ ,
\end{eqnarray*}
whereas on the other hand,
\begin{eqnarray*}
{\cal R}({\cal S}^{\myco}) & = & \underbrace{ \sum_{t \in [n]} \sum_{i \in {\cal C}_{\myavg}} \pr{{\cal S}^{\myco}(t, \myunmark_t, \mymark_t ) = i} \cdot r_i }_{ {\cal R}_{ \myavg }({\cal S}^{\myco}) } + \underbrace{ \sum_{i \in {\cal C}_{\myquit}} \pr{{\cal S}^{\myco}(1, [n], \emptyset) = i} \cdot r_i }_{ {\cal R}_{ \myquit }({\cal S}^{\myco}) }\\
&& \mbox{} + \underbrace{ \sum_{t \in [n+1,2n]} \sum_{i \in {\cal C}_{\mystick}} \pr{{\cal S}^{\myco}(t, \myunmark_t, \mymark_t ) = i} \cdot r_i }_{ {\cal R}_{ \mystick }({\cal S}^{\myco}) } \ .
\end{eqnarray*}
The next few claims compare the above-mentioned terms, eventually leading to
$\rev( {\cal S}^{ \myco } ) \geq (1 - 6\eps) \cdot \rev( {\cal S}^* )$, as desired.

\begin{claim}
${\cal R}_{ \myavg }({\cal S}^{\myco}) = {\cal R}_{ \myavg }({\cal S}^*)$.
\end{claim}
\begin{proof}
To establish this relation, we observe that the expected reward collected by the policy ${\cal S}^{\myco}$ due to average customers identifies with the analogous quantity collected by ${\cal S}^*$, since
\begin{eqnarray*}
{\cal R}_{ \myavg }({\cal S}^{\myco}) & = & \sum_{t \in [n]} \sum_{i \in {\cal C}_{\myavg}} \pr{{\cal S}^{\myco}(t, \myunmark_t, \mymark_t ) = i} \cdot r_i \\
& = & \sum_{t \in [n]} \sum_{i \in {\cal C}_{\myavg}} \pr{{\cal S}^*(t, {\cal A}_t^*) = i} \cdot r_i \\
& = & {\cal R}_{ \myavg }({\cal S}^*) \ ,
\end{eqnarray*}
where the second equality holds since, as previously mentioned, $\myunmark_t$ and ${\cal A}_t^*$ are identically distributed, for every $1 \leq t \leq n$.
\end{proof}

\begin{claim} \label{clm:class_comp_quit}
${\cal R}_{ \myquit }({\cal S}^{\myco}) \geq {\cal R}_{ \myquit }({\cal S}^*) - 2 \eps \cdot {\cal R}({\cal S}^*)$.
\end{claim}
\begin{proof}
In regard to quitters, letting $r_{\max} = \max_{i \in [n]} r_i$ be the maximum reward of any customer, we observe that
\begin{eqnarray*}
{\cal R}_{ \myquit }({\cal S}^*) & = & \sum_{t \in [n]} \sum_{i \in {\cal C}_{\myquit}} \pr{{\cal S}^*(t, {\cal A}_t^*) = i} \cdot r_i \\
& = & \sum_{i \in {\cal C}_{\myquit}} \pr{{\cal S}^*(1, {\cal A}_1^*) = i} \cdot r_i + \sum_{t \in [2,n]} \sum_{i \in {\cal C}_{\myquit}} \pr{{\cal S}^*(t, {\cal A}_t^*) = i} \cdot r_i \\
& \leq & \sum_{i \in {\cal C}_{\myquit}} \pr{{\cal S}^{\myco}(1, [n], \emptyset) = i} \cdot r_i + r_{\max} \cdot \sum_{t \in [2,n]} \sum_{i \in {\cal C}_{\myquit}} \pr{{\cal S}^*(t, {\cal A}_t^*) = i} \\
& \leq & {\cal R}_{ \myquit }({\cal S}^{\myco}) + r_{\max} \cdot \sum_{t \in [2,n]} \sum_{i \in {\cal C}_{\myquit}} \pr{ i \in {\cal A}_t^*} \ ,
\end{eqnarray*}
where the last inequality holds since the event $\{ i \in {\cal A}_t^* \}$, where customer $i$ is still available at stage $t$, clearly contains the event $\{ {\cal S}^*(t, {\cal A}_t^*) = i \}$, in which this customer is served at stage $t$. Consequently, we infer that ${\cal R}_{ \myquit }({\cal S}^{\myco}) \geq {\cal R}_{ \myquit }({\cal S}^*) - 2 \eps \cdot {\cal R}({\cal S}^*)$, by noting that
\begin{eqnarray*}
r_{\max} \cdot \sum_{t \in [2,n]} \sum_{i \in {\cal C}_{\myquit}} \pr{ i \in {\cal A}_t^*} & \leq &
r_{\max} \cdot | {\cal C}_{\myquit} | \cdot \sum_{t \in [2,n]} \left( \frac { \eps }{ n } \right)^{ t-1 } \\
& \leq & r_{\max} \cdot | {\cal C}_{\myquit} | \cdot  \frac { 2\eps }{ n } \\
& \leq & 2\eps \cdot r_{\max} \\
& \leq & 2\eps \cdot {\cal R}({\cal S}^*) \ .
\end{eqnarray*}
Here, the first inequality holds since a necessary condition for customer $i$ to be available at stage $t$ is that s/he has not departed in stages $1, \ldots, t-1$, which occurs with probability $(1 - p_i)^{ t-1 } \leq ( \frac { \eps }{ n } )^{ t-1 }$, as $p_i > 1 - \frac{ \eps }{ n }$ for every $i \in {\cal C}_{\myquit}$. The last inequality is obtained by observing that one possible policy would be to serve the maximum-reward customer at stage $1$, followed by arbitrarily serving customers in subsequent stages. The latter policy clearly has an expected reward of at least $r_{\max}$, and we therefore have ${\cal R}({\cal S}^*) \geq r_{\max}$, due to the optimality of ${\cal S}^*$.
\end{proof}

\begin{claim}
${\cal R}_{ \mystick }({\cal S}^{\myco}) \geq (1 - 4\eps) \cdot {\cal R}_{ \mystick }({\cal S}^*)$.
\end{claim}
\begin{proof}
We first observe that a single sticker does not depart in any of the stages $1, \ldots, 2n$ with probability at least $(1 - \frac{ \eps }{ n^2 })^{ 2n } \geq 1 - \frac{ 4 \eps }{ n }$. Therefore, by the union bound, all stickers do not depart in any of these stages with probability at least $1 - 4 \eps$. Conditional on this event, our policy ${\cal S}^{\myco}$ serves each and every sticker, meaning that
\[ {\cal R}_{ \mystick }({\cal S}^{\myco}) ~~\geq~~ (1 - 4\eps) \cdot \sum_{i \in {\cal C}_{\mystick}} r_i ~~\geq~~ (1 - 4\eps) \cdot {\cal R}_{ \mystick }({\cal S}^*) \ . \]
\end{proof}

\subsection{Proof of Lemma~\ref{lem:reduction_average}} \label{app:reduction}

Recalling that $\rev( {\cal S}^{ \myco } ) \geq (1 - 6\eps) \cdot \rev( {\cal S}^* )$, as shown in  Lemma~\ref{lem:lazy_exists}, it suffices to argue that $\rev( {\cal S} ) \geq (1 - \eps) \cdot \rev( {\cal S}^{\myco} )$. For this purpose, let $\rev_t( S ) = \sum_{i \in [n]} \prpar{ {\cal S}(t, {\cal A}_t) = i }$ be the expected reward collected by our policy ${\cal S}$ at stage $t$. Similarly, $\rev_t( {\cal S}^{\myco} )$ designates the analogous quantity with respect to ${\cal S}^{\myco}$. Noting that both policies operate only along stages $1, \ldots, 2n$, it follows that their total expected rewards can be expressed as $\rev( {\cal S} ) = \sum_{t \in [2n]} \rev_t( {\cal S} )$ and $\rev( {\cal S}^{\myco} ) = \sum_{t \in [2n]} \rev_t( {\cal S}^{\myco} )$. In what follows, we compare these summations by considering three separate collections of stages:
\begin{itemize}
    \item {\em Stage 1}: Here, since ${\cal S}$ duplicates the service decision made by ${\cal S}^{ \myco }$, we clearly have $\rev_1( {\cal S} ) = \rev_1( S^{\myco} )$.
    
    \item {\em Stages $2, \ldots, n$}: Along these stages, ${\cal S}$ first observes the specific realization $A$ of ${\cal A}_2 \cap {\cal C}_\myavg$, and then employs the approximate policy ${\cal S}^{\approx A}$. Therefore,
    \begin{eqnarray*}
    \sum_{t \in [2,n]} \rev_t( {\cal S} ) & = & \sum_{A \subseteq {\cal C}_\myavg} \left( \pr{ {\cal A}_2 \cap {\cal C}_\myavg = A } \cdot \sum_{t \in [2,n]} \rev_t( {\cal S} | {\cal A}_2 \cap {\cal C}_\myavg = A ) \right) \\
    & = & \sum_{A \subseteq {\cal C}_\myavg} \pr{ {\cal A}_2 \cap {\cal C}_\myavg = A } \cdot \rev_{ {\cal I}^A }( {\cal S}^{\approx A} ) \\
    & \geq & (1 - \eps ) \cdot \sum_{A \subseteq {\cal C}_\myavg} \pr{ {\cal A}_2 \cap {\cal C}_\myavg = A } \cdot \rev_{ {\cal I}^A }( {\cal S}^{* A} ) \\
    & \geq & (1 - \eps ) \cdot \sum_{A \subseteq {\cal C}_\myavg} \left( \pr{ {\cal A}_2 \cap {\cal C}_\myavg = A } \cdot \sum_{t \in [2,n]} \rev_t( {\cal S}^{ \myco } | {\cal A}_2 \cap {\cal C}_\myavg = A ) \right) \\
    & = & (1 - \eps ) \cdot \sum_{t \in [2,n]} \rev_t( {\cal S}^{ \myco } ) \ .
    \end{eqnarray*}
    Here, the first inequality holds since $\rev_{ {\cal I}^A }( {\cal S}^{\approx A} ) \geq (1 - \eps) \cdot \rev_{ {\cal I}^A }( {\cal S}^{*A} )$. To obtain the second inequality, note that ${\cal S}^{*A}$ is an optimal policy for ${\cal I}^A$, and therefore, its expected reward is at least as good as that of ${\cal S}^{ \myco }$ for this instance.
    
    \item {\em Stages $n+1, \ldots, 2n$:} In this regime, ${\cal S}$ duplicates the service decisions made by ${\cal S}^{ \myco }$ as well, implying that $\rev_t( {\cal S} ) = \rev_t( S^{\myco} )$ for every $n+1 \leq t \leq 2n$.
\end{itemize}

\subsection{Proof of Lemma~\ref{lem:effect_rnd_v}} \label{app:proof_lem_effect_rnd_v}

\paragraph{First inequality: $\bs{\rev^{ (\tilde{r}) }( {\cal S} ) \geq (1 - \eps) \cdot \rev^{ (r) }( {\cal S} ) - \eps \cdot \rev^{ (r) }( {\cal S}^* )}$.} To obtain a lower bound on $\rev^{ (\tilde{r}) }( {\cal S} ) $, we observe that our reward-rounding method guarantees $\tilde{r}_i \geq \frac{ r_i }{ 1 + \eps } - \frac{ \eps }{ n } \cdot r_{\max}$, for every customer $i \in [n]$. Therefore,
\begin{eqnarray*}
\rev^{ (\tilde{r}) }( {\cal S} ) & = & \sum_{t \in [n]} \sum_{i \in [n]} \pr{{\cal S}(t, {\cal A}_t) = i} \cdot \tilde{r}_i \\
& \geq & \sum_{t \in [n]} \sum_{i \in [n]} \pr{{\cal S}(t, {\cal A}_t) = i} \cdot \left( \frac{ r_i }{ 1 + \eps } - \frac{ \eps }{ n } \cdot r_{\max} \right) \\
& = & \frac{ 1 }{ 1 + \eps } \cdot \rev^{ (r) }( {\cal S} ) - \frac{ \eps }{ n } \cdot r_{\max} \cdot \sum_{i \in [n]} \sum_{t \in [n]}  \pr{{\cal S}(t, {\cal A}_t) = i} \\
& \geq & (1 - \eps) \cdot \rev^{ (r) }( {\cal S} ) - \eps \cdot r_{\max} \\
& \geq & (1 - \eps) \cdot \rev^{ (r) }( {\cal S} ) - \eps \cdot \rev^{ (r) }( {\cal S}^* ) \ .
\end{eqnarray*}
Here, the next-to-last inequality holds since $\sum_{t \in [n]}  \pr{{\cal S}(t, {\cal A}_t) = i} \leq 1$, as this summation is precisely the probability that customer $i$ is served in one of the stages $1, \ldots, n$. The last inequality holds since $\rev^{ (r) }( {\cal S}^* ) \geq r_{\max}$, as explained within the proof of Claim~\ref{clm:class_comp_quit}.

\paragraph{Second inequality: $\bs{\rev^{ (\tilde{r}) }( {\cal S} ) \leq \rev^{ (r) }( {\cal S} )}$.} Conversely, to upper-bound $\rev^{ (\tilde{r}) }( {\cal S} ) $, we observe that our reward-rounding method guarantees $\tilde{r}_i \leq r_i$, for every customer $i \in [n]$, implying that
\begin{eqnarray*}
\rev^{ (\tilde{r}) }( {\cal S} ) & = & \sum_{t \in [n]} \sum_{i \in [n]} \pr{{\cal S}(t, {\cal A}_t) = i} \cdot \tilde{r}_i \\
& \leq & \sum_{t \in [n]} \sum_{i \in [n]} \pr{{\cal S}(t, {\cal A}_t) = i} \cdot r_i \\
& = & \rev^{ (r) }( {\cal S} ) \ . \end{eqnarray*}

\subsection{Proof of Lemma~\ref{lem:comp_prob_short}} \label{app:proof_lem_comp_prob_short}

\paragraph{Case 1: $\bs{p_i \in [\frac{ \eps }{ n^2 }, \frac{ \eps }{ 4 }]}$.} Here, $p^{ \uparrow }_i = U( p_i ) = (1 + \delta) \cdot p^{ \downarrow }_i \in [0,1]$, and therefore
    \begin{eqnarray*}
    ( 1 - p^{ \uparrow }_i )^{ \Delta } & = & ( 1 - (1 + \delta) \cdot p^{ \downarrow }_i )^{ \Delta } \\
    & \geq & e^{ - (1 + \eps) ( 1 + \delta ) p^{ \downarrow }_i \Delta } \\
    & \geq & (1 - p^{ \downarrow }_i)^{ (1 + \eps) ( 1 + \delta )\Delta } \\
    & \geq & (1 - p^{ \downarrow }_i)^{2\eps \Delta} \cdot (1 - p^{ \downarrow }_i)^{ \Delta } \\
    & \geq & \left( 1 - \frac{ \eps }{ 4 } \right)^2 \cdot (1 - p^{ \downarrow }_i)^{ \Delta} \\
    & \geq & \left( 1 - \frac{ \eps }{ 2 } \right) \cdot ( 1 - p^{ \downarrow }_i )^{\Delta} \ .
    \end{eqnarray*}
    To understand the first inequality, note that $1 - x \geq e^{ -(1 + \eps) x }$ for $x \in [0, \frac{ \eps }{ 2 }]$, and we indeed have $(1 + \delta) \cdot p^{ \downarrow }_i \leq \frac{ \eps }{ 2 }$ since $p^{ \downarrow }_i \leq p_i \leq \frac{ \eps }{ 4 }$ and $\delta \leq 1$. The third inequality holds since $\delta = \frac{ \eps^2 }{ 16 }$. The fourth inequality is obtained by recalling that $p^{ \downarrow }_i \leq p_i \leq \frac{ \eps }{ 4 }$ and $\Delta < \frac{ 1 }{ \eps }$.
    
\paragraph{Case 2: $\bs{p_i \in (\frac{ \eps }{ 4 }, 1 - \frac{ \eps }{ n }]}$.} In this case, we set $p^{\uparrow}_i = 1 - D(1 - p_i)$, meaning that $1 - p^{\uparrow}_i = D(1 - p_i) = \frac{ 1 - p^{\downarrow}_i }{ 1 + \delta }$, and therefore
    \begin{eqnarray*}
    ( 1 - p^{\uparrow}_i )^{ \Delta } & = & \left( \frac{ 1 - p^{\downarrow}_i }{ 1 + \delta } \right)^{ \Delta } \\
    & \geq & ( 1 - \delta )^{ 1 / \eps } \cdot (1 - p^{\downarrow}_i)^{ \Delta } \\
    & \geq & e^{-2\delta / \eps } \cdot (1 - p^{\downarrow}_i)^{ \Delta } \\
    & = & e^{-\eps/8 } \cdot (1 - p^{\downarrow}_i)^{ \Delta } \\
    & \geq & \left( 1 - \frac{ \eps }{ 8 } \right) \cdot ( 1 - p^{\downarrow}_i )^{ \Delta } \ ,
    \end{eqnarray*}
    where the first and second inequalities hold since $\Delta < \frac{ 1 }{ \eps }$ and $\delta = \frac{ \eps^2 }{ 16 }$, respectively.

\subsection{Proof of Lemma~\ref{lem:comp_prob_lengthy}} \label{app:proof_lem_comp_prob_lengthy}

\paragraph{Case 1: $\bs{p_i \in [\frac{ \eps }{ n^2 }, \frac{ \eps }{ 4 }]}$.} Here, $p^{ \uparrow }_i = U( p_i ) = (1 + \delta) \cdot p^{ \downarrow }_i \in [0,1]$, and therefore
    \begin{eqnarray*}
    ( 1 - p^{ \uparrow }_i )^{(1 - \eps) \Delta } & = & ( 1 - ( 1 + \delta ) \cdot p^{ \downarrow }_i )^{(1 - \eps) \Delta } \\
    & \geq & e^{ -(1 + \eps) ( 1 + \delta ) (1 - \eps) p^{ \downarrow }_i \Delta } \\
    & \geq & e^{-p^{ \downarrow }_i \Delta} \\
    & \geq & ( 1 - p^{ \downarrow }_i )^{ \Delta } \ .
    \end{eqnarray*}
    To understand the first inequality, note that $1 - x \geq e^{ -(1 + \eps) x }$ for $x \in [0, \frac{ \eps }{ 2 }]$, and we indeed have $( 1 + \delta ) \cdot p^{ \downarrow }_i \leq \frac{ \eps }{ 2 }$ since $p^{ \downarrow }_i \leq p_i \leq \frac{ \eps }{ 4 }$ and $\delta \leq 1$. The second inequality follows by recalling that $\delta = \frac{ \eps^2 }{ 16 }$.
    
\paragraph{Case 2: $\bs{p_i \in (\frac{ \eps }{ 4 }, 1 - \frac{ \eps }{ n }]}$.} In this case, we set $p^{ \uparrow }_i = 1 - D(1 - p_i)$, meaning that $1 - p^{ \uparrow }_i = D(1 - p_i) = \frac{ 1 - p^{ \downarrow }_i }{ 1 + \delta }$, and therefore
    \begin{eqnarray*}
    ( 1 - p^{ \uparrow }_i )^{(1 - \eps) \Delta } & = & \left( \frac{ 1 - p^{ \downarrow }_i }{ 1 + \delta } \right)^{(1 - \eps) \Delta } \\
    & \geq & ( 1 - \delta )^{ \Delta } \cdot (1 - p^{ \downarrow }_i)^{-\eps \Delta } \cdot ( 1 - p^{ \downarrow }_i )^{ \Delta } \\
    & \geq & e^{ -2 \delta \Delta } \cdot e^{ \eps p^{ \downarrow }_i \Delta } \cdot ( 1 - p^{ \downarrow }_i )^{\Delta} \\
    & \geq &( 1 - p^{ \downarrow }_i )^{\Delta} \ ,
    \end{eqnarray*}
    where the last inequality holds since
    \[ \eps p^{ \downarrow }_i ~~=~~ \eps \cdot ( p^{ \uparrow }_i - \delta \cdot (1 - p^{ \uparrow }_i) ) ~~\geq~~ \eps \cdot (p_i - \delta) ~~\geq~~ \eps \cdot \left( \frac{ \eps }{ 4 } - \frac{ \eps^2 }{ 16 } \right) ~~\geq~~ \frac{ 3\eps^2 }{ 16 } ~~>~~ 2\delta \ . \]
    
\subsection{Proof of Theorem~\ref{thm:rev_diff_easy}} \label{app:proof_thm_rev_diff_easy}

\paragraph{Constructing $\bs{{\cal S}^{\downarrow}}$.} For ease of analysis, in addition to keeping track of the collection of available customers ${\cal A}_t$ at the beginning of any stage $t$, we will allow our  service policy ${\cal S}^{ \downarrow }$ to maintain an additional state parameter ${\cal E}_t$, which is referred to as the current collection of ``eliminated'' customers. Initially, all customers are uneliminated, meaning that ${\cal E}_1 = \emptyset$. Then, in each stage $t \in [n]$, we create ${\cal E}_{t+1}$ by augmenting ${\cal E}_t$ with a random set of customers, produced by independently picking every customer $i \in [n]$ with probability $\frac{ p^+_i - p^-_i }{ 1 - p^-_i }$. It is not difficult to verify that such extended policies can easily be translated back to our standard notion of an adaptive service policy with no deterioration in their expected reward.

Now, in order to define the policy ${\cal S}^{\downarrow}$, when the latter arrives at state $(t,{\cal A}_t, {\cal E}_t)$, it serves precisely the same customer served by ${\cal S}$ when all eliminated customers are ignored, meaning that ${\cal S}^{\downarrow} (t,{\cal A}_t, {\cal E}_t) = {\cal S} (t,{\cal A}_t \setminus {\cal E}_t)$.

\paragraph{Analysis.} For the purpose of relating between the expected rewards $\rev^{ (p^-) }( {\cal S}^{ \downarrow } )$ and $\rev^{ (p^+) }( {\cal S} )$, we first show that the random sets ${\cal A}_t^{{\cal S}^{\downarrow}} \setminus {\cal E}_t^{{\cal S}^{\downarrow}}$ and ${\cal A}_t^{\cal S}$, respectively corresponding to the policies ${\cal S}^{\downarrow}$ and ${\cal S}$, are identically distributed.

\begin{lemma} \label{lem:same_prob_elim}
$\prpar{ {\cal A}_t^{{\cal S}^{\downarrow}} \setminus {\cal E}_t^{{\cal S}^{\downarrow}} = A } = \prpar{ {\cal A}_t^{\cal S} = A }$, for every $t\geq 1$ and $A \subseteq [n]$.
\end{lemma}
\begin{proof}
The proof works by induction on $t$, noting that the base case of $t = 1$ trivially holds, since ${\cal A}_1^{{\cal S}^{\downarrow}} = [n]$, ${\cal E}_1^{{\cal S}^{\downarrow}} = \emptyset$, and ${\cal A}_t^{\cal S} = [n]$ with probability $1$. For the general case of $t \geq 2$, we have
\begin{eqnarray*}
\pr{ {\cal A}_t^{\cal S} = A } & = & \sum_{ \MyAbove{ B \supsetneq A : }{ {\cal S}(t-1,B) \notin A} } \pr{ {\cal A}_{t-1}^{\cal S} = B } \cdot \pr{ {\cal A}_t^{\cal S} = A | {\cal A}_{t-1}^{\cal S} = B } \\
& = & \sum_{ \MyAbove{ B \supsetneq A : }{ {\cal S}(t-1,B) \notin A} } \pr{ {\cal A}_{t-1}^{\cal S} = B } \cdot \prod_{i \in B \setminus (A \cup \{ {\cal S}(t-1,B) \})} p_i^+ \cdot \prod_{i \in A} (1 - p^+_i) \\
& = & \sum_{ \MyAbove{ B \supsetneq A : }{ {\cal S}(t-1,B) \notin A} } \prpar{ {\cal A}_{t-1}^{{\cal S}^{\downarrow}} \setminus {\cal E}_{t-1}^{{\cal S}^{\downarrow}} = B } \\
& & \qquad \qquad \cdot \prod_{i \in B \setminus (A \cup \{ {\cal S}(t-1,B) \})} \underbrace{ \left( 1 - (1 - p^-_i ) \cdot \left(1 - \frac{ p^+_i - p^-_i }{ 1 - p^-_i } \right) \right) }_{ p_i^+ } \\
& & \qquad \qquad \cdot \prod_{i \in A} \underbrace{ \left( (1 - p^-_i ) \cdot \left( 1 - \frac{ p^+_i - p^-_i }{ 1 - p^-_i } \right) \right) }_{ 1 - p_i^+ }\\
& = & \sum_{ \MyAbove{ B \supsetneq A : }{ {\cal S}(t-1,B) \notin A} } \prpar{ {\cal A}_{t-1}^{{\cal S}^{\downarrow}} \setminus {\cal E}_{t-1}^{{\cal S}^{\downarrow}} = B } \cdot \pr{ {\cal A}_{t}^{{\cal S}^{\downarrow}} \setminus {\cal E}_{t}^{{\cal S}^{\downarrow}} = A \left| {\cal A}_{t-1}^{{\cal S}^{\downarrow}} \setminus {\cal E}_{t-1}^{{\cal S}^{\downarrow}} = B \right. } \\
& = & \pr{ {\cal A}_t^{{\cal S}^{\downarrow}} \setminus {\cal E}_t^{{\cal S}^{\downarrow}} = A } \ ,
\end{eqnarray*}
where the third equality follows from the induction hypothesis.
\end{proof}

Given this result, we immediately obtain the desired relation, by noting that
\begin{eqnarray*}
\rev^{ (p^-) }( {\cal S}^{ \downarrow } ) & = & \sum_{t \geq 1} \sum_{A \subseteq [n]} \pr{ {\cal A}_t^{{\cal S}^{\downarrow}} \setminus {\cal E}_t^{{\cal S}^{\downarrow}} = A } \cdot r_{ {\cal S}^{\downarrow} (t,{\cal A}_t^{{\cal S}^{\downarrow}}, {\cal E}_t^{{\cal S}^{\downarrow}}) | {\cal A}_t^{{\cal S}^{\downarrow}} \setminus {\cal E}_t^{{\cal S}^{\downarrow}} = A } \\
& = & \sum_{t \geq 1} \sum_{A \subseteq [n]} \pr{ {\cal A}_t^{\cal S} = A } \cdot r_{ {\cal S} (t,A) } \\
& = & \rev^{ (p^+) }( {\cal S} ) \ .
\end{eqnarray*}

\subsection{Proof of Lemma~\ref{lem:no_future}} \label{app:proof_lem_no_future}

Let us focus on a single customer $i\in [n]$ and on any preceding stage, $\tau_0 \leq \tau$. To avoid duplicated contents, we utilize the extra piece of notation $\bar{\tau}_0$, such that $\bar{\tau}_0 = \tau_0$ when $\tau_0$ is a regular stage, and such that $\bar{\tau}_0 = \tau_0 - 1$ when $\tau_0$ is a milestone. By closely inspecting how $Y^\downarrow_{i,\tau_0 }$ is defined in Section~\ref{subsec:coupling}, it follows that we are actually employing deterministic mappings $\pi_1(\cdot)$ and $\pi_2(\cdot)$, corresponding to cases~1 and~2, such that
\[ Y^\downarrow_{i,\tau_0 } ~~=~~  \begin{cases}
    \pi_1(W_{i,\tau_0},{V}_{i,\tau_0}, {{\cal Y}^\downarrow_{i}}^{\leq\tau_0-1},{{\cal Y}^\uparrow_{i}}^{\leq\mu^{-1}(\tau_0 -1)}) \qquad & \text{if } \tau_0 \in {\cal M}_{\gamma} \\
    \pi_2(Y^\uparrow_{i,\mu^{-1}(\tau_0 )},Z_{i,\tau_0 }) & \text{if } \tau_0 \notin {\cal M}_{\gamma}
\end{cases} \]
By leveraging these equations for every $\tau_0 \leq \tau$, it immediately follows that we can specify deterministic mappings $\nu_1(\cdot)$ and $\nu_2(\cdot)$ such that 
\[ {{\cal Y}^{\downarrow}}^{\leq \tau} ~~=~~ 
\begin{dcases}
 \nu_1({\cal Z}^{\leq \tau},{\cal W}^{\leq \tau},{\cal V}^{\leq \tau},{{\cal Y}^{\uparrow}}^{\leq  \mu^{-1}(\tau)}) \qquad & \text{if } \tau\notin {\cal M}_\gamma \\
 \nu({\cal Z}^{\leq \tau},{\cal W}^{\leq \tau},{\cal V}^{\leq \tau},{{\cal Y}^{\uparrow}}^{\leq  \mu^{-1}(\tau-1)}) \qquad & \text{if } \tau\in {\cal M}_\gamma
\end{dcases} \]
In light of our construction, we have enforced that $({\cal Z}^{>\tau},{\cal W}^{>\tau},{\cal V}^{>\tau},{\cal Y}^{\uparrow >  \mu^{-1}(\bar{\tau})})$ is independent of $({\cal Z}^{\leq \tau },{\cal W}^{\leq \tau },{V}^{\leq \tau },{\cal Y}^{\uparrow\leq\mu^{-1}(\bar{\tau})})$, and therefore, the equations above imply that ${{\cal Y}^{\downarrow}}^{\leq \tau} $ is independent of $({\cal Z}^{>\tau},{\cal W}^{>\tau},{\cal V}^{>\tau},{\cal Y}^{\uparrow >  \mu^{-1}(\bar{\tau})})$.

\end{document}